\documentclass[twocolumn]{autart}    
\let\theoremstyle\relax
\usepackage{amsthm}
\usepackage{amssymb}
\usepackage{mathtools}
\usepackage{amsmath}
{
	\newtheorem{definition}{Definition}
	
	\newtheorem{theorem}{Theorem}
	
	\newtheorem{remark}{Remark}
	\newtheorem{corollary}{Corollary}
	\newtheorem{lemma}{Lemma}
	\newtheorem{proposition}{Proposition}
	
}
\usepackage{enumerate}
\usepackage{amsfonts}
\usepackage{xcolor}
\usepackage{graphicx}          

\begin{document}
\graphicspath{{Plots/}}
\begin{frontmatter}

\title{On Attack Detection and Identification for the Cyber-Physical System using Lifted System Model
} 

\thanks[footnoteinfo]{This research is funded by the Secure Systems Research Center (SSRC) at Technology Innovation Institute (TII), UAE. The authors are grateful to Dr. Shreekant (Ticky) Thakkar and his team members at the SSRC for their valuable comments and support.}
\author[]{Dawei Sun}\ead{sun289@purdue.edu},    
\author[]{Minhyun Cho}\ead{cho515@purdue.edu}, and
\author[]{Inseok Hwang}\ead{ihwang@purdue.edu}  
%


\begin{abstract}                          
Motivated by the safety and security issues related to cyber-physical systems with potentially multi-rate, delayed, and nonuniformly sampled measurements, we investigate the attack detection and identification using the lifted system model in this paper. Attack detectability and identifiability based on the lifted system model are formally defined and rigorously characterized in a novel approach. The method of checking detectability is discussed, and a residual design problem for attack detection is formulated in a general way. For attack identification, we define and characterize it by generalizing the concept of mode discernibility for switched systems, and a method for identifying the attack is discussed based on the theoretical analysis. An illustrative example of an unmanned aircraft system (UAS) is provided to validate the main results. 
\end{abstract}

\end{frontmatter}

\section{Introduction} \label{sec:intro}
Cyber-physical systems (CPS) can be operated effectively by interweaving computation and communication resources with the physical process \cite{pasqualetti2015control}, but the increasing complexity and the close interaction between the logical and physical components make the CPS vulnerable to cyber layer threats \cite{fillatre2017security}. For example, the GPS receiver is vulnerable to jamming \cite{grant2009gps}, meaconing \cite{zhu2022novel}, and spoofing attacks \cite{siamak2020dynamic}, while the communication channel without proper encryption and authentication can be exploited for man-in-the-middle attacks \cite{conti2016survey} or denial-of-service attacks \cite{pelechrinis2010denial}, which can cause detrimental impact in the physical world. Different from the problem of fault detection and isolation (FDI), the diagnosis and containment of cyber-layer attacks are more challenging, considering that the attack, by using the information from eavesdropping or side-channel attacks, can be elaborated to penetrate cyber-layer protection and deceive the physical-layer monitoring system. It induces the concept of \textit{stealthy attack} \cite{kwon2017reachability} and motivates the study of attack detectability and identifiability issues \cite{pasqualetti2013attack}. 

The stealthy attack has been investigated from the control-theoretic perspective. For example, the vulnerability of unmanned aircraft systems (UAS) to stealthy attack has been discussed in \cite{kwon2014analysis}; the stealthy attack design problem has been studied in \cite{li2020worst,zhang2019optimal}; in \cite{bai2017data}, it is pointed out that the stealthiness of attack can be independent of the use of monitoring system; in \cite{mo2013detecting}, the stealthiness of replay attack is investigated; in \cite{park2019stealthy}, it is shown that the stealthiness of the attack can be designed to be robust to modeling uncertainties; last but not least, the authors in \cite{sui2020vulnerability} provide a decent definition and characterization of the vulnerability of dynamic systems to stealthy attacks. For preventing stealthy attacks, cryptography is typically considered to protect communication channels, but heavy-weight encryption techniques may not be applicable to systems with limited communication bandwidth or computational power \cite{jovanov2019relaxing,khazraei2020attack,miao2016coding,shang2020optimal}. Recently, there are also some control-theoretic approaches proposed as complement for preventing falsified sensing information from deceiving the system, e.g., additive/multiplicative/switched watermarking schemes \cite{mo2015physical,ferrari2017detection,ferrari2020switching}, and bating/moving target approaches \cite{tian2019moving,weerakkody2016moving,griffioen2020moving}, which introduce artificial uncertainties so that the attacker cannot have sufficient information to design stealthy attacks (similar ideas have also been considered in \cite{mao2020novel,ding2022application}). Unfortunately, these approaches still lack extensive theoretical analysis, and their effectiveness relies on the assumption that the attacker does not know these approaches are utilized, which may not be the case when the attacker can perform eavesdropping or side-channel attacks.  

Note that both the aforementioned information-security and control-theoretic approaches have their advantages and limitations. To firmly assure the system safety and security, \textit{analytic redundancies} \cite{hwang2009survey}, a concept raised in FDI, are necessary, since all model-based detection algorithms are essentially comparing the observed behavior of the system with that of the model which describes the nominal behavior. A widely considered way to realize the theoretical redundancies is hardware redundancies, e.g., using redundant sensors, as it is well-known that the data injection attack to sensing information can be detected if it is ``sparse” \cite{pasqualetti2013attack}. However, in practical implementation, blindly adding redundant sensors may not be effective. Similar sensors can be correlated so that they might be affected by the attacks simultaneously \cite{wood2002denial}, which do not contribute to analytic redundancies. Adding too many redundant sensors can also increase the complexity of the system as well as the running cost, which degenerates the system performance \cite{vitus2012efficient}. Furthermore, the additional sensors for monitoring may have different features: they can have different sampling rates, delays, accuracies, and/or vulnerabilities to attacks. It is nontrivial to investigate the effective way to fuse the sensing information with different characteristics for attack detection and identification.  

In this paper, we non-conservatively define and rigorously characterize attack detectability and identifiability for the CPS with potentially multi-rate, delayed, and/or nonuniformly sampled measurements, which has important contributions for analyzing the CPS's analytic redundancies. For such a CPS, the lifting techniques \cite{meyer1990new,zhang2002fault,chow1984analytical} have been widely considered for state estimation \cite{li2008kalman,shen2022multi} and fault detection problems \cite{li2006subspace,zhang2016fault,izadi2007analysis,zhong2007observer}. Although the lifted system can preserve the same input-output behavior of the original system, the state of the lifted system may not have the same practical meanings, so the existing works (e.g., \cite{sui2020vulnerability}) on the study of system resilience cannot be applied by trivial extensions. In this paper, we define and characterize attack detectability by using the lifting technique and generalizing the analysis in \cite{sui2020vulnerability}, which is followed by discussions on the methods of checking or assuring detectability and the residual design problem for attack detection. For attack identification, motivated by the multiple model approaches \cite{menke1995sensor,shima2002efficient,tanwani2010inversion} (i.e., the compromised system behavior subject to each type of attack can be associated with a mode of the switched system), we define and characterize attack identifiability by relaxing the concept of mode discernibility of switched systems subject to unknown input \cite{gomez2011observability,sun2021controlled,sun2022controlled,fiore2017secure,boukhobza2011observability}, which is the property for whether the mode can be correctly recovered from the measurement. Based on it, a framework of attack identification method with assured properties can be derived.

The rest of the paper is organized as follows. In Section~\ref{sec:PF}, the definitions of attack detectability and identifiability based on the lifted system model are provided. Some lemmas for deriving the main results are presented in Section~\ref{sec:Prelim}. The main results on attack detection and identification are derived and discussed in Sections~\ref{sec:Det} and \ref{sec:Idt}, respectively. In Section \ref{sec:NE}, an UAS example is presented to demonstrate the theoretical results. Finally, Section~\ref{sec:Conclusion} concludes this paper.

\section{Problem Formulation} \label{sec:PF}
In this section, we revisit the lifting system technique for describing the input-to-output behavior of the cyber-physical system (CPS) with potentially multi-rate, delayed, and/or nonuniformly sampled measurement, and then formally define attack detectability and identifiability based on it.
\subsection{Lifting Technique for Cyber-Physical Systems}
In this paper, we consider the discrete-time linear time-invariant system model to describe the nominal input-to-state behavior of the CPS:
\begin{equation} \label{eq:NominalSys}
\hat x_{t+1} = \hat A \hat x_t + \hat B^u \hat u_t,
\end{equation}
where $\hat x_t \in \mathbb{R}^n$ is the state vector of the nominal CPS at time $t$ which may include the states with physical significance as well as the observer or dynamic controller states, and $\hat u_t \in \mathbb{R}^{m^u}$ is the nominal reference input to the CPS. As the nominal behavior of a CPS should be stable, the eigenvalues of $\hat A$ are assumed to be located within the unit disk of the complex plane. For the multi-rate measurements with potential delay and nonuniform sampling, it is assumed that there is a \textit{frame period} \cite{li2006subspace,li2008kalman,zhong2007observer}, $T$, such that for each sensor $i \in \{1,..., I\}$, there are $M_i$ samples obtained by the end of each frame $\{kT,...,(k+1)T-1\}$. These samples are related to the states and the input within that frame: 
\begin{equation}
\left\{
\begin{aligned}
\hat y^i_{kT+t^i_1} &= \hat C_i \hat x_{kT+t^i_1} + \hat D^u_i  \hat u_{kT+t^i_1} \\
\hat y^i_{kT+t^i_2} &= \hat C_i \hat x_{kT+t^i_2} + \hat D^u_i \hat u_{kT+t^i_2} \\
&\vdots \\
\hat y^i_{kT+t^i_{M_i}} &= \hat C_i \hat x_{kT+t^i_{M_i}} + \hat D^u_i \hat u_{kT+t^i_{M_i}} \\
\end{aligned}
\right.,
\end{equation}
which is illustrated in Fig.~\ref{fig:FramePeriod}. 

\begin{figure}[htb]
	\centering
	\includegraphics[trim = 0mm 0mm 0mm 0mm, clip, scale=0.25]{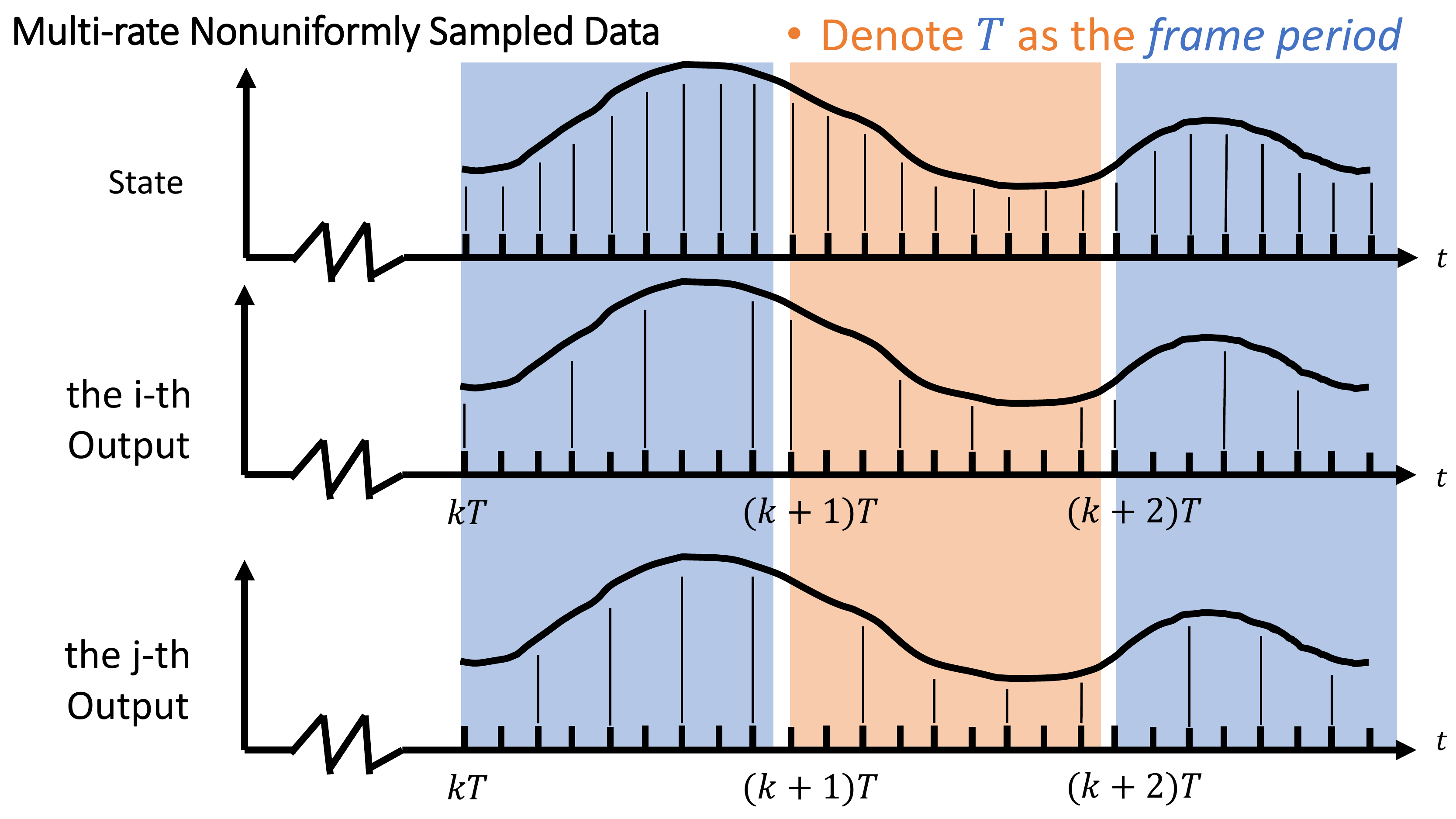}
	\caption{Illustration of Frame Period}
	\label{fig:FramePeriod}
\end{figure}

By introducing the following notations:
\begin{equation}
\begin{aligned}
\hat x_k &= \hat x_{kT}, \\
\hat u_k &= \begin{bmatrix}
\hat u_{kT}{}', \dots, \hat u_{(k+1)T-1}{}'
\end{bmatrix}', \\
\hat y_k &=	\begin{bmatrix}
\hat y^1_{kT+t^1_1} \dots \hat y^1_{kT+t^1_{M_1}} \dots \hat y^I_{kT+t^I_1} \dots \hat y^I_{kT+t^I_{M_I}}
\end{bmatrix}',
\end{aligned}	
\end{equation}
where $'$ denotes the transpose, the nominal behavior of the system can be described by the lifted system model:
\begin{equation} \label{eq:NomLiftedSys}
\hat \Sigma: \left\{
\begin{aligned}
\hat x_{k+1} &= A \hat x_k + B^u \hat u_k \\
\hat y_k &= C \hat x_k + D^u \hat u_k
\end{aligned}
\right. ,
\end{equation}
where $A$, $B^u$, $C$, $D^u$ matrices can be constructed using the standard approach given in \cite{li2006subspace,li2008kalman,zhong2007observer}: 
\begin{equation*} 
	\begin{aligned}
	A &= \hat A^T, &B^u = \begin{bsmallmatrix}
	\hat A^{T-1} \hat B^u &\hat A^{T-2} \hat B^u &\dots &\hat A \hat B^u &\hat B^u
	\end{bsmallmatrix},& \\
	C &= \begin{bsmallmatrix}
	\hat C_1 \hat A^{t^1_1} \\
	\vdots\\
	\hat C_I \hat A^{t^I_{M_I}}
	\end{bsmallmatrix}, &D^u = \begin{bsmallmatrix}
	\hat C_1 \hat A^{t^1_1-1} \hat B^u &\dots \,\ \hat D^u_1 &\dots &0\\
	\vdots &\ddots &\vdots  &\vdots \\
	\hat C_I \hat A^{t^I_{M_I}-1} \hat B^u &\dots &\hat D^u_I &\dots 0\\
	\end{bsmallmatrix},&
	\end{aligned}
\end{equation*}
where $T$ is the aforementioned frame period.

Different from the nominal system, the actual system states are subject to disturbance and noise, and it could be potentially compromised by the data injection type of attacks:
\begin{equation}
\begin{aligned}
	x^a_{t+1} = \hat A x^a_t + \hat B^u \hat u_t + \hat B^a_{q^a} a_t + \hat B^w w_t,
\end{aligned}
\end{equation}
where $x^a_t$ is the actual system state vector, $a_t$ is the false data injection signal, and $w_t$ represents the bounded disturbance and noise. Note that $q^a$ is the attack mode that takes values in $Q$, the set of candidate attack modes. $\hat B^a_{q^a}$ stands for the effect of attack mode $q^a$ to the system, while $\hat B^w$ is the channel from the disturbance and noise to the states. Note that the data injection attack can be realized by spoofing the sensors, so the measurements are potentially vulnerable to attacks:
\begin{equation}
\begin{aligned}
y^{a,i}_{kT+t^i_j} = \hat C_i x^a_{kT+t^i_j} &+ \hat D^u_i \hat u_{kT+t^i_j} \\
&+ \hat D^a_{i,q^a} a_{kT+t^i_j} +  \hat D^w_{i} w_{kT+t^i_j},\\
\end{aligned}
\end{equation}
where $y^{a,i}_{kT+t^i_j}$ is the falsified $j$-th measurement from the $i$-th sensor in the $k$-th frame, and matrices $\hat D^a_{i,q^a}$ and $\hat D^w_{i}$ represent how this measurement is affected by the attack associated with mode $q^a$ and noise, respectively. Similarly, by introducing the following notations:
\begin{equation}
\begin{aligned}
x^a_k &= x^{a}_{kT}, \\
a_k &= \begin{bmatrix}
a_{kT}{}', \dots, a_{(k+1)T-1}{}'
\end{bmatrix}', \\
w_k &= \begin{bmatrix}
w_{kT}{}', \dots, w_{(k+1)T-1}{}'
\end{bmatrix}',\\
y^a_k &= \begin{bmatrix}
y^{a,1}_{kT+t^1_1} \dots \hat y^{a,1}_{kT+t^1_{M_1}} \dots \hat y^{a,I}_{kT+t^I_1} \dots \hat y^{a,I}_{kT+t^I_{M_I}}
\end{bmatrix}',
\end{aligned}	
\end{equation}
the actual system subject to disturbance, noise, and attack can be described by the lifted system $\Sigma^a$:
\begin{equation} \label{eq:ActLiftedSys}
\Sigma^a: \left\{
\begin{aligned}
x^a_{k+1} &= A x^a_k + B^u \hat u_k + B^a_{q^a} a_k + B^w w_k\\
y^a_k &= C x^a_k + D^u \hat u_k + D^a_{q^a} a_k + D^w w_k
\end{aligned}
\right. ,
\end{equation}
where $B^a_{q^a}$, $B^w$, $D^a_{q^a}$, and $D^w$ can be constructed using the similar approach for constructing $B^u$ and $D^u$.

\subsection{Attack Detectability and Identifiability}
With the nominal system $\hat \Sigma$ in (\ref{eq:NomLiftedSys}) and the actual system $\Sigma^a$ in (\ref{eq:ActLiftedSys}), we can further define
\begin{equation}
	x_k = x^a_k - \hat x_k, \,\ y_k = y^a_k - \hat y_k.
\end{equation}
Note that $x_k$ reveals how the state vector of $\Sigma^a$ deviates from that of $\hat \Sigma$, and $y_k$ represents the difference between the observed measurements from the actual system and the expected measurements from the nominal system. It should be remarked that $y_k$ is closely related to the attack detection property: if $||y_k||$ is identically zero or sufficiently small, there is no way to distinguish the impact of the attack or effect of the disturbance or noise. On the other hand, if $||y_k||$ becomes sufficiently large, then the measured output significantly deviates from the expected output from the nominal system so that the monitoring system can trigger the alarm. For this reason, $||y_k||$ can stand for the \textit{stealthiness} of the attack.

Meanwhile, it could be improper to just use $||x_k||$ to measure the impact of the attack. First, $x_k$ only represents the deviation of the system state vector at certain time instances. It is possible that $||x_k|| = ||x^a_{kT} - \hat x_{kT}||$ is small for all $k$, but $||x^a_t - \hat x_t||$ becomes large at some $t$, which has been demonstrated in \cite{kim2016zero}. Second, it is not necessary that all elements of $x_k$ are safety-critical. As $x_k$ may include the dummy states or states that do not have physical significance, whether those states deviate from the nominal states does not necessarily reflect the \textit{severity} of the attack. For these reasons, we introduce a new variable $z_t$ to extract and scale the deviation of safety-critical states:
\begin{equation}
\begin{aligned}
		z_t &= \hat E x_t, \\
		z_k &= \begin{bmatrix}
		z_{kT}{}', \dots, z_{(k+1)T-1}{}'
		\end{bmatrix}'. \\
\end{aligned}
\end{equation}
Using $\hat \Sigma$ and $\Sigma^a$, we have the following system which characterizes the relationships between $x_k$, $y_k$ and $z_k$:
\begin{equation} \label{eq:DeltaLiftedSys}
\Sigma^\Delta: \left\{
\begin{aligned}
x_{k+1} &= A x_k + B^a_{q^a} a_k + B^w w_k\\
y_k &= C x_k + D^a_{q^a} a_k + D^w w_k \\
z_k &= E x_k + F^a_{q^a} a_k + F^w w_k
\end{aligned}
\right. ,
\end{equation}
where $E$, $F^a$ and $F^w$ can be constructed using the similar approach for constructing $C$, $D^a$ and $D^w$. For convenience, we use the following notations to denote the input-to-state and input-to-ouput behaviors of $\Sigma^\Delta$:
\begin{equation}\label{eq:ISIOmaps}
\begin{aligned}
	x_k &= \mathbf{x}^{q^a}_k(x_0, \{a_k\}_{k=0}^\infty, \{w_k\}_{k=0}^\infty), \\
	y_k &= \mathbf{y}^{q^a}_k(x_0, \{a_k\}_{k=0}^\infty, \{w_k\}_{k=0}^\infty), \\
	z_k &= \mathbf{z}^{q^a}_k(x_0, \{a_k\}_{k=0}^\infty, \{w_k\}_{k=0}^\infty).
\end{aligned}
\end{equation}
The explicit expressions of these mapping are well-known and thus omitted. Note that for each $K$, $x_K$, $y_K$, and $z_K$ do not depend on $\{a_k\}_{k = K+1}^\infty$ by causality, so $\{a_k\}_{k=0}^\infty$ in (\ref{eq:ISIOmaps}) can be either a sequence or replaced with a finite truncation of it without causing confusions. 

With the above notations, we now formally define attack detectability:
\begin{definition}[Detectability] \label{def:Detectability}
	The attack mode $q^a \in Q$ is \textbf{detectable} if there are $\delta > 0$ and $\epsilon > 0$ such that 
	\begin{equation}
		||\mathbf{z}^{q^a}_{k_z}(0, \{a_k\}_{k=0}^\infty, 0)|| \geq \delta \,\ \text{for some } k_z
	\end{equation}
	implies 
	\begin{equation}
		||\mathbf{y}^{q^a}_{k_y}(0, \{a_k\}_{k=0}^\infty, 0)|| \geq \epsilon \,\ \text{for some } k_y.
	\end{equation}
	If the attack mode $q^a$ is not detectable, we say the CPS is \textbf{vulnerable} to attack mode $q^a$.
\end{definition}
In words, the attack mode $q^a$ is called detectable if there is no severe attack preventing detection under this mode: whenever the attack causes a severe impact on the system, i.e., $||z_k||$ becomes large, the attack can be revealed by comparing the observed output and the output from the nominal system. Equivalently, for the system that is not vulnerable to the attack mode $q^a$, if the attack remains ``$\epsilon$-stealthy" (i.e., $||y_k||\leq \epsilon$ for all $k$), $||z_k||$ must be bounded by $\delta$. Indeed, detectability defined here is sufficient and necessary for detecting any severe attacks, though $x_0$ and $\{w_k\}_k$ are set to zeros in the definition, which will be discussed in Section \ref{sec:Det}.
\\
\begin{remark}
	A similar way for defining detectability or vulnerability has been considered in \cite{sui2020vulnerability}, which provides so far the most general result to the best of our knowledge. However, the severity of the attack in that work is related to the deviation of the state vector, which could be improper when the lifted system is considered, as discussed. Our approach for characterizing detectability is different from the approach in \cite{sui2020vulnerability}, and our results are more general. 
\end{remark}

For attack identifiability, we provide the formal definition as:
\begin{definition}[Identifiability] \label{def:Identifiability}
	A pair of modes $p \neq q \in Q$ are called \textbf{discernible}, if there are $\delta > 0$ and $\epsilon > 0$ such that either
	\begin{equation}
	||\mathbf{z}^{p}_{k_z}(x^p_0, \{a^p_k\}_{k=0}^\infty, 0)|| \geq \delta, 
	\end{equation}
	or 
	\begin{equation}
	||\mathbf{z}^{q}_{k_z}(x^q_0, \{a^q_k\}_{k=0}^\infty, 0)|| \geq \delta, 
	\end{equation}
	for some $k_z$ implies
	\begin{equation}
	||\mathbf{y}^{p}_{k_y}(x^p_0, \{a^p_k\}_{k=0}^\infty, 0) - \mathbf{y}^{q}_{k_y}(x^q_0, \{a^q_k\}_{k=0}^\infty, 0)|| \geq \epsilon
	\end{equation}
	for some $k_y$. If any pair of modes $p \neq q$ in $Q$ are discernible, the set of attack modes, $Q$, is called \textbf{identifiable}. 
\end{definition}
In words, discernibility means that the two modes result in sufficiently different output $\{y_k\}_k$ whenever one of them causes sufficiently large impact on the system. If there is a pair of modes that are indiscernible (i.e., not discernible), a severe impact might be caused while there is no way to distinguish these two modes, which means that an attack can be designed to cause severe impact on the system while preventing the attack identification. The definition as proposed is sufficient and necessary to identify the attack mode whenever the attack can potentially cause severe impact on the system. Although $\{w_k\}_k$ is set to zero in the definition, this definition is applicable for the cases where $\{w_k\}_k$ is nonzero but bounded, which will be discussed in Section \ref{sec:Idt}.  
\\
\begin{remark}
	In the proposed definition, we relate attack identifiability with the concept of mode discernibility (or mode distinguishability, mode observability) of the switched system with unknown input \cite{gomez2011observability,sun2021controlled,sun2022controlled}. However, for mode discernibility of the switched system with unknown input, two modes are called discernible if the outputs from the two modes are different no matter what unknown inputs are injected into the system, which is rather restrictive in general. Our definition relaxes the requirement in the sense that two modes do not have to behave differently unless a severe impact is potentially caused by the attack, which has not been considered in existing works. 
\end{remark}

In the rest of the paper, we will find necessary and sufficient conditions for detectability and identifiability, which will be followed by the discussions on the detection and identification schemes. 

\section{Preliminaries} \label{sec:Prelim}
The study of attack detectability, revealed by its definition, is closely related to the \textit{zero dynamics} \cite{de2001geometric,isidori2013zero}: vulnerability means  $\{z_k\}_k$ can be made arbitrarily large by some $\{a_k\}_k$ while $\{y_k\}_k$ can remain zero or arbitrarily small. In addition, by introducing the augmented system, the discernibility can also be related to output-nulling problems \cite{de2011location,gomez2011observability}. For this reason, we revisit the concept and notation of the maximum output-nulling subspace so that detectability can be strictly characterized. 

Let us consider system $\Sigma^\Delta$ in (\ref{eq:DeltaLiftedSys}) with a given fixed $q^a$. The maximal $y$-nulling subspace of this system is denoted by $\mathcal{V}$, which is well-known to be the maximal subspace satisfying that there is a matrix $M$ such that
\begin{equation}\label{eq:VisInv}
\begin{aligned}
(A+B^a_{q^a}M)\mathcal{V} &\subset \mathcal{V}, \\
(C+D^a_{q^a}M)\mathcal{V} &= \{0\}.
\end{aligned}
\end{equation}
In general, the choice of $M$ is not unique. Suppose $N$ is a matrix satisfying
\begin{equation}\label{eq:invDOF}
\begin{aligned}
\text{Im}\{N\} &\subset \text{ker}\{D^a_{q^a}\}, \\
\text{Im}\{B^a_{q^a}N\} &= B^a_{q^a} \text{ker}\{D^a_{q^a}\} \cap \mathcal{V},
\end{aligned}
\end{equation}
where we use $\text{Im}$ and $\text{ker}$ to denote the column space and null space, respectively. Then, for any $\bar M$ taking the form $M+NK$ ($K$ is arbitrary), it also satisfies (\ref{eq:VisInv}). For convenience, we say $(M,N)$ is a \textit{friend} of $\mathcal{V}$, denoted by $(M,N) \in \mathcal{F}(\mathcal{V})$, if $M$ and $N$ satisfy (\ref{eq:VisInv}) and (\ref{eq:invDOF}), respectively. Note that a state is contained in $\mathcal{V}$ if and only if starting from this state, the output $y$ can be made identically zero by some input $\{a_k\}_k$. For the rest of this paper, we assume $\mathcal{V}$ is not trivially $\{0\}$ without loss of generality.

The following lemma presents the relationship between any two pairs of friends of $\mathcal{V}$, which can be considered as a way to parameterize $\mathcal{F}(\mathcal{V})$.
\\
\begin{lemma} \label{lm:Friends}
	Suppose $(\hat M, \hat N) \in \mathcal{F}(\mathcal{V})$ and $(\bar M, \bar N) \in \mathcal{F}(\mathcal{V})$. Then, for any $V$ whose columns span $\mathcal{V}$,
	\begin{equation}
	B^a_{q^a}(\hat M - \bar M)V = B^a_{q^a} \hat{N} \hat{K}
	\end{equation}
	for some $\hat K$ with the proper dimension, and 
	\begin{equation}
	(\hat M - \bar M)V = \hat N \hat K + \hat H_M
	\end{equation}
	for some $\hat H_M$ so that $\text{Im}\{\hat H_M\} \in \text{ker}\{B^a_{q^a}\}\cap \text{ker}\{D^a_{q^a}\}$. Besides, we have
	\begin{equation}
	\bar N = \hat N \hat L + \hat{H}_N
	\end{equation}
	for some matrix $\hat L$ with the proper dimension and for some $\hat H_N$ satisfying $\text{Im}\{\hat H_N\} \in \text{ker}\{B^a_{q^a}\}\cap \text{ker}\{D^a_{q^a}\}$.
\end{lemma}

\begin{proof}
	See Appendix A.
\end{proof}

Note that $(M,N) \in \mathcal{F}(\mathcal{V})$ is useful to parameterize output-nulling input sequences as well as the output-nulling dynamics, which is shown in the following lemma:
\\
\begin{lemma}[$y$-nulling Dynamics] \label{lm:NullingPolicy}
	For system $\Sigma^\Delta$, suppose $(M,N)$ is a friend of $\mathcal{V}$. A state $x_0$ and a sequence $\{a_k\}_k$ satisfy
	\begin{equation}
	\mathbf{y}^{q^a}_k(x_0, \{a_k\}_k,0) = 0, \forall k \in \mathbb{N}
	\end{equation}
	if and only if $x_0 \in \mathcal{V}$ and for every $k$, $a_k$ can be written as 
	\begin{equation}\label{eq:nullingpolicy}
	a_k = M x_k + N \tilde{a}_k
	\end{equation}
	for some $\tilde{a}_k$, where $\{x_k\}_k$ is the state trajectory resulted from $x_0$ and $\{a_k\}_k$.
\end{lemma}

\begin{proof}
	See Appendix B.
\end{proof}

Note that (\ref{eq:VisInv}) shows that $\mathcal{V}$ is $(A+BM)$-invariant. It is well-known that in general, a subspace $\mathcal{S}$ is $A$-invariant, i.e., $A\mathcal{S} \subset \mathcal{S}$, if and only if for any $S$ whose columns form a basis of $\mathcal{S}$, there is a matrix $A|_S$ satisfying $AS = SA|_S$. As $A$ maps $S$ to itself, $A|_S$ is called $A$ restricted in $\mathcal{S}$. The eigenvalues and (generalized) eigenvectors of these two matrices are related according to the following well-known lemma.
\\
\begin{lemma} \label{lm:InvEig}
	If $\lambda$ is an eigenvalue of $A|_S$, it is also an eigenvalue of $A$. $v$ is a eigenvector of $A|_S$ associated with $\lambda$ if and only if $Sv$ is an eigenvector of $A$ associated with the same eigenvalue.	Similarly, $J(\lambda)$ and $G$ are the Jordan block and the corresponding chain of generalized eigenvectors of $A|_S$ if and only if $J(\lambda)$ and $SG$ satisfy $ASG = SGJ(\lambda). $
\end{lemma}

The next lemma can formalize the idea that if a vector is ``close'' to the two subspaces, then this vector is ``close'' to the intersection of these two subspaces. The following lemma is adapted from \cite{piziak1999constructing}:
\\
\begin{lemma}\label{lm:projection}
	Let $\mathcal{V}$ and $\mathcal{W}$ be two subspaces, and let $P_\mathcal{V}$ and $P_\mathcal{W}$ be the orthogonal projection matrices onto $\mathcal{V}$ and $\mathcal{W}$, respectively. Then, the orthogonal projection matrix onto $\mathcal{V} + \mathcal{W}$, $P_{\mathcal{V} + \mathcal{W}}$, can be obtained by 
	\begin{equation}
	P_{\mathcal{V} + \mathcal{W}} = (P_\mathcal{V}+P_\mathcal{W})^\dagger(P_\mathcal{V} + P_\mathcal{W}).
	\end{equation}
\end{lemma}

From this lemma, we can see that $P_{(\mathcal{V}\cap \mathcal{W})}^\perp \triangleq P_{(\mathcal{V}\cap \mathcal{W})^\perp} = P_{\mathcal{V}^\perp + \mathcal{W}^\perp} = (P_\mathcal{V}^\perp + P_\mathcal{W}^\perp)^\dagger (P_\mathcal{V}^\perp + P_\mathcal{W}^\perp)$. Note that we use $\dagger$" to denote the pseudoinverse. 

\section{Attack Detection} \label{sec:Det}
In this section, we derive and discuss the main results on attack detection using the lifted system.
\subsection{Characterization of Vulnerability}
To characterize attack detectability in Definition~\ref{def:Detectability}, we further denote the $(A,B^a_{q^a})$-controllable subspace by $\mathcal{C}$, and define $\mathcal{V}^*$ to be $\mathcal{C}\cap \mathcal{V}$. Note that $\mathcal{V}^*$ also has the invariance property as $\mathcal{V}$, i.e., if a matrix $M$ satisfies (\ref{eq:VisInv}), then $(A+B^a_{q^a}M)\mathcal{V}^* \subset \mathcal{V}^*$ and $(C+D^a_{q^a}M)\mathcal{V}^* = \{0\}$. Let $V^*$ be a matrix whose columns form the basis of $\mathcal{V}^*$, then for a matrix $M$ satisfying (\ref{eq:VisInv}), $(A+B^a_{q^a}M)|_{V^*}$ is well defined such that $(A+B^a_{q^a}M)V^* = V^*(A+B^a_{q^a}M)|_{V^*}$. With these notations and lemmas in Section \ref{sec:Prelim}, we now present the characterization of attack detectability.
\\
\begin{theorem}\label{thm:Detectability}
	Consider the system described by (\ref{eq:NomLiftedSys}), (\ref{eq:ActLiftedSys}), and (\ref{eq:DeltaLiftedSys}), it is vulnerable to attack mode $q^a$ if and only if at least one of the following conditions holds:
	\begin{enumerate}[(i)]
	\item $\text{ker}\{B^a_{q^a}\}\cap \text{ker}\{D^a_{q^a}\} \not\subset \text{ker}\{F^a_{q^a}\}$;
	\item for some $(\hat M,\hat N) \in \mathcal{F}(\mathcal{V})$, $ L_{(\hat M, \hat N)}^{\Sigma^\Delta} \neq 0$, where $L_{(\hat M, \hat N)}^{\Sigma^\Delta}$ is defined as
	\begin{equation}
	\begin{aligned}
	 \Big[(E+F^a_{q^a}\hat{M})(A&+B^a_{q^a}\hat{M})^{n-1}B^a_{q^a}\hat{N} \,\ \dots \,\  \\
	          &(E+F^a_{q^a}\hat{M})B^a_{q^a}\hat N \,\ F^a_{q^a}\hat N \Big];
	\end{aligned}
	\end{equation}
	\item for some $(\hat M,\hat N) \in \mathcal{F}(\mathcal{V})$, $(A+B^a_{q^a}\hat{M})|_{V^*}$ has a unstable eigenvalue $\lambda$ (i.e., $|\lambda| \geq 1$) whose generalized eigenspace $\mathcal{S}$ satisfies \footnote{When the eigenvalue and the generalized eigenvectors are complex, the condition becomes that either the real part or imaginary part of one of the generalized eigenvector is not contained in $\text{ker}\{(E+F^a_{q^a}\hat M)V^*\}$.} 
	\begin{equation}\label{eq:EigNotKer}
	V^*\mathcal{S} \not\subset \text{ker}\{(E+F^a_{q^a}\hat{M})\}.
	\end{equation}
	\end{enumerate}
\end{theorem} 
  
\begin{proof}
	The proof is given in Appendix C.
\end{proof}

The proof of Theorem~\ref{thm:Detectability} is challenging, but we can gain important insights from it. When the system is vulnerable, the proof of Theorem 1 provides explicit ways to design the stealthy and severe attack. If condition (i) holds, there is an attack input sequence $\{a_k\}_k$ such that it affects $\{z_k\}_k$ while resulting in identically zero $\{x_k\}_k$ and $\{y_k\}_k$, assuming the initial condition $x_0 = 0$. By scaling this $\{a_k\}_k$, a stealthy and sufficiently severe attack can be obtained. If condition (ii) holds, the output-nulling dynamics of $\Sigma^\Delta$ has the nontrivial controllable subspace, which means the attack policy in the form of (\ref{eq:nullingpolicy}) can be used to generate stealthy and severe attacks. If condition (iii) holds, the attacker can activate an unstable eigenvector of $(A+B^a_{q^a}\hat M)$ so that the attack policy $a_k = \hat M x_k$ can generate a stealthy and severe attack. 

When the attack mode $q^a$ is detectable, the proof gives a relationship between $\epsilon$ and $\delta$. It should be remarked that their relationship is crucial for designing and analyzing the attack detection using $||y_k||$. 
Recall that by detectability, attack severity $||z_{k_z}|| \geq \delta$ for some $k_z$ implies $||y_{k_y}|| \geq \epsilon$ for some $k_y$, so $\delta$ can represent the sensitivity for attack detection when $\epsilon$ is used as the threshold for triggering the alarm (i.e., the alarm is triggered when $||y_k|| > \epsilon$). Equivalently,  the attack stealthiness $||y_{k}|| < \epsilon$ for all $k$ implies attack severity $||z_k|| < \delta$, so $\delta$ can also be interpreted as the bound on the impact caused by undetected attacks. Although $\delta(\epsilon)$ obtained in the proof could be conservatively over-approximated depending on the choice of $(\hat M,\hat N) \in \mathcal{F}(\mathcal{V})$ in (\ref{eq:rewriteBa}), the conservativeness might be reduced if a $\mathtt{z}$-domain approach is applied based on the ideas in \cite{meyer1990new,zhong2007observer,sui2020vulnerability}. 
For example, consider the case where the attack input $\{a_k\}_k$ and outputs $\{y_k\}_k$, $\{z_k\}_k$ of system $\Sigma^\Delta$ in (\ref{eq:DeltaLiftedSys}) have $\mathtt{z}$-transforms $a_\mathtt{z}$, $y_\mathtt{z}$, and $z_\mathtt{z}$, respectively. Then, in the $\mathtt{z}$-domain, the following relationship hold when $x_0 = 0$: 
\begin{equation}
	y_\mathtt{z} = T_y (\mathtt{z}) a_\mathtt{z}, \,\ z_\mathtt{z} = T_z(\mathtt{z}) a_\mathtt{z},
\end{equation}
where $T_y (\mathtt{z})$ and $T_z (\mathtt{z})$ are $a$-to-$y$ and $a$-to-$z$ transfer function matrices, respectively. Suppose $T_y$ can be written in the Smith-McMillan form $P_L S_y P_R $, and $T_y^\dagger$ given by $P_R^{-1} S_y^\dagger P_L^{-1}$ is stable. Then, it can be shown that $z_\mathtt{z} = T_z(\mathtt{z}) T_y^\dagger(\mathtt{z}) y_\mathtt{z}$ under the condition of detectability, which implies that $\delta(\epsilon)$ can be taken as $||T_z(\mathtt{z}) T_y^\dagger(\mathtt{z})||_1 \epsilon$, where $||T_z(\mathtt{z}) T_y^\dagger(\mathtt{z})||_1 = \Sigma ||R^{zy}_k||_2$ and $||R^{zy}_k||_2$ is the spectrum norm of the impulse response matrix of $T_z(\mathtt{z}) T_y^\dagger(\mathtt{z})$ at $k$. 

\subsection{Assuring Attack Detectability}
We have found the necessary and sufficient conditions for vulnerability in Theorem~\ref{thm:Detectability}, but it is also critical to discuss how the conditions (i)-(iii) in the theorem can be checked. 

It is straightforward to check whether condition (i) in Theorem~\ref{thm:Detectability} holds or not, but to check conditions (ii) and (iii), we need to use some properties of $\mathcal{V}$ and $\mathcal{F}(\mathcal{V})$. Fortunately, we have the following proposition showing if condition (i) does not hold, we can take an arbitrary $(\hat M,\hat N) \in \mathcal{F}(\mathcal{V})$ and check whether $L_{(\hat M, \hat N)}^{\Sigma^\Delta} = 0$ to determine whether condition (ii) holds: if we find that $L_{(\hat M, \hat N)}^{\Sigma^\Delta} = 0$, then all other $(M,N) \in \mathcal{F}(\mathcal{V})$ will make $L_{(M, N)}^{\Sigma^\Delta}$ zero. 
\\
\begin{proposition}\label{pr:GeometricMN}
	Suppose condition (i) in Theorem~\ref{thm:Detectability} does not hold. Then condition (ii) holds if and only if $L_{(M, N)}^{\Sigma^\Delta} \neq 0$ for all $(M, N) \in \mathcal{F}(\mathcal{V})$.
\end{proposition}
The proof is omitted as it can be derived from Lemma~\ref{lm:Friends}. Checking if condition (iii) in Theorem~\ref{thm:Detectability} holds or not is more complicated. We need the following lemma which parameterizes $\{(A+B^a_{q^a}\hat{M})|_{V^*} | (\hat M,\hat N) \in \mathcal{F}(\mathcal{V})\}$:
\\
\begin{lemma}\label{lm:ParA+BM}
	Suppose $V^*$ is the matrix whose columns form a basis of $\mathcal{V}^*$ (recall that $\mathcal{V}^* = \mathcal{V} \cap \mathcal{C}$). Let $(\hat M, \hat N)$ be a friend of $\mathcal{V}$ and $\hat B_N$ be a matrix satisfying $B^a_{q^a}\hat N = V^*\hat B_N$. Then, for any $(M,N) \in \mathcal{F}(\mathcal{V})$, there exist some $\hat K$ such that
	\begin{equation}\label{eq:A+BM2A+BMhat}
	(A+B^a_{q^a}M)|_{V^*} = (A+B^a_{q^a} \hat M)|_{V^*} + \hat{B}_N \hat K.
	\end{equation} 
	On the other hand, for any $\hat K$, there is a matrix $M$ such that $(M,\hat N) \in \mathcal{F}(\mathcal{V})$ and (\ref{eq:A+BM2A+BMhat}) holds. 
\end{lemma}
The proof is omitted as it is a result from Lemma~\ref{lm:Friends}. This lemma shows that for any $M$ satisfying (\ref{eq:VisInv}), $(A+B^a_{q^a}M)|_{V^*}$ can be considered as a ``closed-loop" system matrix where $(A+B^a_{q^a} \hat M)|_{V^*}$ is the ``open-loop" system matrix and $\hat K$ is the ``feedback gain". According to this relationship, we can see that for any $K$, all the generalized eigenvectors of $(A+B^a_{q^a}\hat M)|_{V^*} + \hat B_N K)$ corresponding to the $((A+B^a_{q^a}\hat M)|_{V^*}, \hat B_N)$-controllable eigenvalues are always included in the $((A+B^a_{q^a}\hat M)|_{V^*}, \hat B_N)$-controllable subspace. 
\\
\begin{proposition}\label{lm:Kunstableuncontrollable}
	Suppose conditions (i) and (ii) in Theorem~\ref{thm:Detectability} do not hold. Condition (iii) in the theorem holds if and only if $(A+B^a_{q^a}\hat M)|_{V^*}$ has an unstable eigenvalue that is not $((A+B^a_{q^a}\hat M)|_{V^*}, \hat B_N)$-controllable, and there is a ``gain matrix'' $K$ such that one of the eigenvectors or generalized eigenvectors of $((A+B^a_{q^a}\hat M)|_{V^*} + \hat B_N K)$ associated with this eigenvalue is not contained in the null space of $(E+F^a_{q^a}\hat M) V^*$.
\end{proposition}

\begin{proof}
	See Appendix D.
\end{proof}

This proposition states that if conditions (i) and (ii) of Theorem~\ref{thm:Detectability} do not hold, condition (iii) holds if and only if a certain modification of the eigenspace associated with an uncontrollable \footnote{By ``controllable", we mean $((A+B^a_{q^a}\hat M)|_{V^*}, \hat B_N)$-controllable throughout this subsection.} and unstable eigenvalue of $(A+B^a_{q^a}\hat M)|_{V^*}$ is possible via the ``feedback gain" $K$. The important thing revealed by the proposition is that we only need to focus on the uncontrollable and unstable eigenvalues \footnote{Indeed, these eigenvalues are invariant zeros}
when we are checking whether there exists a qualified $K$. To see if such a $K$ exists or not, one can refer to the conventional problem of eigenspace assignment in \cite{sinswat1977eigenvalue}, where it has been shown that there is a $K$ such that $\lambda$ and $G$ are an eigenvalue and some corresponding generalized eigenvectors (not necessarily all) of $(A+B^a_{q^a}\hat M)|_{V^*}+\hat B_N K$ if and only if
\begin{equation}\label{eq:SylvesterTypeEq}
(I - \hat B_N \hat{B}_N^\dagger) (GJ(\lambda) - (A+B^a_{q^a}\hat M)|_{V^*} G) = 0,
\end{equation}
where $J(\lambda)$ is the Jordan block associated with $\lambda$ and its size should be compatible with $G$. If $\lambda$ is an uncontrollable eigenvalue, its multiplicity in the uncontrollable subsystem is invariant no matter what $K$ is chosen. Thus, one can start with finding the solution space of (\ref{eq:SylvesterTypeEq}) assuming $J$ has the size of $1$-by-$1$, and if there is a solution not contained in $\text{ker}\{(E+F^a_{q^a}\hat M)V^*\}$, one can claim the system is vulnerable. Otherwise, one can increase the size of $J$. Furthermore, one can also check other unstable uncontrollable eigenvalues by repeating the process. It can be claimed that the attack mode $q^a$ is detectable if and only if for any uncontrollable unstable eigenvalue $\lambda$ with feasible $J(\lambda)$, the solution space of (\ref{eq:SylvesterTypeEq}) is always contained in $\text{ker}\{(E+F^a_{q^a}\hat M)V^*\}$. Last but not least, it should be pointed out that equation (\ref{eq:SylvesterTypeEq}) is merely a linear equation where $G$ is the unknown, so one can at least characterize its solution space by rewriting it into the standard form. Besides, readers can refer to general Sylvester equations. The following two remarks provide two special cases where condition (iii) of Theorem~\ref{thm:Detectability} can be checked easily. 
\\
\begin{remark}
	Consider a special case that $\text{ker}\{(E+F^a_{q^a}\hat M)V^*\}$ is trivially $\{0\}$. It is implied that $\hat B_N$ is zero when coniditon (ii) in Theorem~\ref{thm:Detectability} does not hold. It means none of the eigenvalues of $(A+B^a_{q^a}\hat M)|_{V^*}$ is controllable, and so, the condition in Proposition~\ref{lm:Kunstableuncontrollable} is equivalent to the fact that $(A+B^a_{q^a}\hat M)|_{V^*}$ is unstable, which is coincident with the main result in \cite{sui2020vulnerability}.\\
\end{remark}

\begin{remark}
	Consider another special case where $\text{ker}\{(E+F^a_{q^a}\hat M)V^*\}$ is strictly equal to the controllable subspace. For this case, we can show the condition in Proposition~\ref{lm:Kunstableuncontrollable} is equivalent to the fact that $(A+B^a_{q^a}\hat M)|_{V^*}$ has an unstable eigenvalue whose generalized eigenspace is not contained in $\text{ker}\{(E+F^a_{q^a}\hat M)V^*\}$. 
\end{remark}

\subsection{Attack Detection under Bounded Uncertainties}
In this section, we look into an attack detection strategy:
\begin{equation}
\begin{aligned}
||y_k|| \geq \epsilon &\implies \text{ Alarm is triggered at } k. \\
\end{aligned}	
\end{equation}
According to Definition \ref{def:Detectability}, if the attack mode $q^a$ is detectable, an alarm will be triggered if $||z_{k_z}|| > \delta(\epsilon)$ at some $k_z$, assuming $x_0 = 0$ and $\{w_k\}_k = 0$. Meanwhile, with the same setup, there is no false alarm no matter what $\epsilon > 0$ is taken, i.e., $||y_k||$ is identically zero when $\{a_k\}_k = 0$. It is certainly not the case where $\{w_k\}_k$ is bounded but nonzero. In this case, the false alarm can be avoided while the sufficiently severe attack can be detected if $\epsilon$ is well-designed, which is shown in the following corollary:
\\
\begin{proposition}\label{pr:detectionwithnoise}
Suppose the attack mode $q^a$ is detectable, $\hat \Sigma$ in (\ref{eq:NomLiftedSys}) is stable and $||w_k|| \leq 1, \forall k$, without loss of generality. There are $\delta$ and $\epsilon$ such that 
\begin{enumerate}[(i)]
\item There is no false alarm, i.e., \\
$||\mathbf{y}^{q^a}_{k}(0, 0, \{w_k\}_{k=0}^\infty)|| < \epsilon$ for all $k \in \mathbb{N}$; and 
\item Severe attacks can be detected, i.e., \\
$||\mathbf{y}^{q^a}_{k_y}(0, \{a_k\}_{k=0}^\infty, \{w_k\}_{k=0}^\infty)|| \geq \epsilon$ for some $k_y$ if \\ $||\mathbf{z}^{q^a}_{k_z}(0, \{a_k\}_{k=0}^\infty, \{w_k\}_{k=0}^\infty)|| \geq \delta$ for some $k_z$. 
\end{enumerate}
\end{proposition}

\begin{proof}
	By the stability condition, $\epsilon$ can be taken as a number greater than $\Sigma ||R^{yw}_k||_2$ to satisfy statement (i), where $R^{yw}_k$ is the impulse response matrix of (\ref{eq:DeltaLiftedSys}) from $w_k$ to $y_k$. Taking $\tilde{\delta}$ as $\Sigma ||R^{wz}_k||_2$, where $R^{zw}_k$ is the impulse response matrix from $w_k$ to $z_k$, we have $||\mathbf{z}^{q^a}_{k}(0, 0, \{w_k\}_k)|| < \tilde \delta$ for all $k$. By detectability, there is a $\bar \delta (2\epsilon)$ such that  $||\mathbf{y}^{q^a}_{k}(0, \{a_k\}_{k=0}^\infty, 0)|| < 2 \epsilon$ for all $k$ implies $||\mathbf{z}^{q^a}_{k}(0, \{a_k\}_{k=0}^\infty, 0)|| < \bar \delta$ for all $k$. Let $\delta$ be $\bar{\delta} + \tilde{\delta}$. Now, we can see that if $||\mathbf{y}^{q^a}_{k_y}(0, \{a_k\}_{k=0}^\infty, \{w_k\}_{k=0}^\infty)|| < \epsilon$ for all $k$, then,
	\begin{equation}
	\begin{aligned}
	&||\mathbf{y}^{q^a}_{k}(0, \{a_k\}_{k=0}^\infty, 0)|| \\
	\leq &||\mathbf{y}^{q^a}_{k}(0, \{a_k\}_{k=0}^\infty, \{w_k\}_{k=0}^\infty)|| + ||\mathbf{y}^{q^a}_{k}(0, 0, \{-w_k\}_{k=0}^\infty)|| \\
	< &2 \epsilon,
	\end{aligned}
	\end{equation}
	for all $k$, which implies
	\begin{equation}
	\begin{aligned}
	&||\mathbf{z}^{q^a}_{k}(0, \{a_k\}_{k=0}^\infty, \{w_k\}_{k=0}^\infty)|| \\
	\leq &||\mathbf{z}^{q^a}_{k}(0, \{a_k\}_{k=0}^\infty, 0)|| + ||\mathbf{z}^{q^a}_{k}(0, 0, \{w_k\}_{k=0}^\infty)|| \\
	< &\bar{\delta} + \tilde{\delta} = \delta
	\end{aligned}
	\end{equation}
	for all $k$. Therefore, $\delta$ and $\epsilon$ satisfy statement (ii). 
\end{proof}

\begin{remark}\label{rm:designresidual}
	Detectability remains invariant when the output $y$ is transformed to $\tilde{y}_\mathtt{z} = H_y(\mathtt{z})y_\mathtt{z}$, where $H_y(\mathtt{z})$ is any stable and invertible filter with a stable inverse. Note that with the transformed output, the relationship between $\delta$ and $\epsilon$ can be different, which means by choosing an appropriate $H_y(\mathtt{z})$, $\tilde{y}$ could be more sensitive to the attack $\{a_k\}_k$ instead of noise $\{w_k\}_k$. The similar problems have been widely considered in the area of fault detection and isolation. Motivated by Proposition~\ref{pr:detectionwithnoise}, we propose a general formulation for this design problem: 
	\begin{equation}
		\min_{\delta, \epsilon, H_y(\mathtt{z})} \delta
	\end{equation}
	subject to \\
	(1) $H_y(\mathtt{z})$ is stable and invertible with stable inverse;\\
	and for any $\{a_k\}_k$ and $\{w_k\}_k$ ($||w_k||\leq 1$, $\forall k$),\\
	(2) $||\tilde{\mathbf{y}}^{q^a}_{k}(0, 0, \{w_k\}_{k=0}^\infty)|| < \epsilon$ for all $k$; \\
	(3) $||\tilde{\mathbf{y}}^{q^a}_{k}(0, \{a_k\}_{k=0}^\infty, \{w_k\}_{k=0}^\infty)|| < \epsilon$ for all $k$ \\ $\implies$ 
	$||\mathbf{z}^{q^a}_{k_z}(0, \{a_k\}_{k=0}^\infty, \{w_k\}_{k=0}^\infty)|| < \delta$ for all $k$. \\
	Note that we use $\tilde{\mathbf{y}}^{q^a}_{k}(\cdot, \cdot, \cdot)$ to denote the input-to-output behavior of $\tilde{y}$ that is obtained by $\tilde{y}_\mathtt{z} = H_y(\mathtt{z}) y_\mathtt{z}$. The problem is feasible according to Proposition~\ref{pr:detectionwithnoise}, and it can be relaxed and solved easily, if needed, by assuming a certain structure of $H_y(\mathtt{z})$ and replacing (3) with a conservative but simpler relationship between $\delta$ and $\epsilon$. By solving this design problem, $\tilde{y}$ can be used to monitor the system, which guarantees the sensitivity to the attack quantified by $\delta$ as well as the absence of false alarm. We would remark that this formulation also has a limitation that the time required to detect the severe attack is not guaranteed, and addressing this issue can be a future work. 
\end{remark}

\section{Attack Identification} \label{sec:Idt}
In this section, we present the main results on attack identification using the lifted system.
\subsection{Characterization of Discernibility}
Motivated by existing works on mode discernibility for switched systems, we can construct the augmented system that facilitates the investigation:
\begin{equation} \label{eq:AugSys}
\Sigma^\Delta_{pq}: \left\{
\begin{aligned}
x^{pq}_{k+1} &= A_{pq} x^{pq}_k + B^a_{pq} a^{pq}_k\\
y^{pq}_k &= C_{pq} x^{pq}_k + D^a_{pq} a^{pq}_k\\
z^{pq}_k &= E_{pq} x^{pq}_k + F^a_{pq} a^{pq}_k
\end{aligned}
\right. ,
\end{equation}
where $x^{pq}_k = [x^p_k{}' \,\ x^q_k{}']'$, $z^{pq}_k = [z^p_k{}' \,\ z^q_k{}']'$, $a^{pq}_k = [a^p_k{}' \,\ a^q_k{}']'$, $y^{pq}_k = y^p_k - y^q_k$, and
\begin{equation*}
\begin{aligned}
A_{pq}  &= \begin{bsmallmatrix}
A &0 \\ 0 &A
\end{bsmallmatrix}, 
&B^a_{pq}  &= \begin{bsmallmatrix}
B^a_p &0 \\ 0 &B^a_q
\end{bsmallmatrix}, \\
C_{pq} &= \begin{bsmallmatrix}
-C &C
\end{bsmallmatrix}, &D^a_{pq}  &= \begin{bsmallmatrix}
D^a_p &-D^a_q
\end{bsmallmatrix},\\
E_{pq}  &= \begin{bsmallmatrix}
E &0 \\ 0 &E
\end{bsmallmatrix}, 
&F^a_{pq}  &= \begin{bsmallmatrix}
F^a_p &0 \\ 0 &F^a_q
\end{bsmallmatrix}, \\
\end{aligned}
\end{equation*}
It can be seen that by $\Sigma^\Delta_{pq}$, discernibility is related to the output-nulling problem, similar to detectability that we have characterized in Section~\ref{sec:Det}. For characterization of discernibility, we use $\mathcal{V}_{pq}$ to denote the maximal $y^{pq}$-nulling subspace, and $(M,N) \in \mathcal{F}(\mathcal{V}_{pq})$ to indicate that they satisfy
\begin{equation}\label{eq:VpqisInv}
\begin{aligned}
(A_{pq}+B^a_{pq}M)\mathcal{V}_{pq} &\subset \mathcal{V}_{pq}, \\
(C_{pq}+D^a_{pq}M)\mathcal{V}_{pq} &= \{0\},
\end{aligned}
\end{equation}
and
\begin{equation}\label{eq:VpqinvDOF}
\begin{aligned}
\text{Im}\{N\} &\subset \text{ker}\{D^a_{pq}\} \\
\text{Im}\{B^a_{pq}N\} &= B^a_{pq} \text{ker}\{D^a_{pq}\} \cap \mathcal{V}_{pq}.
\end{aligned}
\end{equation}
In addition, we present the input-to-state and input-to-output behaviors of $\Sigma^\Delta_{pq}$ as
\begin{equation}\label{eq:ISIOAugmaps}
\begin{aligned}
x^{pq}_k &= \mathbf{x}^{pq}_k(x^{pq}_0, \{a^{pq}_k\}_{k=0}^\infty) \\
y^{pq}_k &= \mathbf{y}^{pq}_k(x^{pq}_0, \{a^{pq}_k\}_{k=0}^\infty) \\
&= \mathbf{y}^{p}_k(x^{p}_0, \{a^{p}_k\}_{k=0}^\infty,0) - \mathbf{y}^{q}_k(x^{q}_0, \{a^{q}_k\}_{k=0}^\infty, 0),\\
z^{pq}_k &= \mathbf{z}^{pq}_k(x^{pq}_0, \{a^{pq}_k\}_{k=0}^\infty)\\
&= \begin{bmatrix}
\mathbf{z}^{p}_k(x^{p}_0, \{a^{p}_k\}_{k=0}^\infty,0)\\
\mathbf{z}^{q}_k(x^{q}_0, \{a^{q}_k\}_{k=0}^\infty,0)
\end{bmatrix}
\end{aligned}
\end{equation}
for convenience. With these notations, we present the main result on the characterization of discernibility now:
\\
\begin{theorem}\label{thm:Discernibility}
	A pair of attack modes $p$ and $q$ from $Q$ ($p \neq q$)  are discernible if and only if all of the following conditions hold:
	\begin{enumerate}[(i)]
	\item $\text{ker}\{B^a_{pq}\}\cap \text{ker}\{D^a_{pq}\} \subset \text{ker}\{F^a_{pq}\}$;
	\item for any $(M, N) \in \mathcal{F}(\mathcal{V}_{pq})$, 
	\begin{equation}
	\begin{aligned}
		F^a_{pq} N &= 0, \\
		(E_{pq}+F^a_{pq}M)\mathcal{V}_{pq} &= \{0\}.
	\end{aligned}
	\end{equation}
\end{enumerate}
\end{theorem}

\begin{proof}
	See Appendix E.
\end{proof}

Note that by taking any arbitrary $(\hat M, \hat N)\in \mathcal{F}(\mathcal{V}_{pq})$, condition (ii) holds if and only if $F^a_{pq} \hat N = 0$ and $(E_{pq}+F^a_{pq}\hat M)\mathcal{V}_{pq} = \{0\}$, similar to the statements in Proposition~\ref{pr:GeometricMN}. To assure that the set of attack modes $Q$ is identifiable, we need to check discernibility between every pair of modes in $Q$. 

\subsection{Attack Mode Identification}

The following corollary derived from Theorem~\ref{thm:Discernibility} shows that if a sufficiently large impact is caused at time instant $k_0$ by the attack mode $p$, then mode $q$, distinguishable from mode $p$, must behave differently within $n+1$ time steps (recall that $n$ is the dimension of the state space).
\\
\begin{corollary}\label{cr:n+1stepID}
	The following statements are equivalent:
	\begin{enumerate}[(i)]
		\item Attack modes $p$ and $q$ are discernible;
		\item There exist $\epsilon > 0$ and $\delta > 0$ such that 
		\begin{equation}
		||\mathbf{z}^{pq}_{k_0}(x^{pq}_{k_0}, \{a^{pq}_k\}_{k=k_0}^\infty)|| \geq \delta
		\end{equation}
		implies
		\begin{equation}
		||\mathbf{y}^{pq}_k(x^{pq}_{k_0}, \{a^{pq}_k\}_{k=k_0}^\infty)|| \geq \epsilon
		\end{equation}
		for some $k \in \{k_0, ..., k_0 + n + 1\}$.
	\end{enumerate}
\end{corollary}
\begin{proof}
	See Appendix F.
\end{proof} 

This corollary implies that a projection-based method can discern the attack modes. Suppose for each mode $q$, at time instance $k_0+n+1$, the following residual can be computed by using the observed measurements $Y_{k_0:k_0+n+1} = [y_{k_0}{}', \,\ ... \,\ y_{k_0+n+1}{}']'$ (the similar way has been considered in the previous works, e.g., \cite{alessandri2005receding}):
\begin{equation}\label{eq:defResidual}
	r^q_{k_0} = ||P^\perp_{\mathcal{Y}^q_{n+1}} Y_{k_0:k_0+n+1}||,
\end{equation}
where $P^\perp_{\mathcal{Y}^q_{n+1}}$ is the projection matrix onto the orthogonal complement subspace of ${\mathcal{Y}^q_{n+1}} \triangleq \text{Im}\{[\mathcal{O}^q_{n+1}\,\ \mathcal{G}^q_{n+1}]\}$, and $\mathcal{O}^q_{n+1}$ and $\mathcal{G}^q_{n+1}$ are given as 
\begin{equation*}
\begin{aligned}
\mathcal{O}^q_{n+1} = \begin{bsmallmatrix}
C  \\
C A \\
\vdots\\
C A^{n+1}
\end{bsmallmatrix}, \mathcal{G}^q_{n+1} = \begin{bsmallmatrix}
D^a_{q} &0 &\cdots &0 \\
CB^a_{q} &D^a_{q} &\cdots &0 \\
\vdots &\vdots &\ddots &\vdots \\
CA^{n}B^a_{q} &CA^{n-1}B^a_{q} &\cdots &D^a_{q}
\end{bsmallmatrix}.
\end{aligned}
\end{equation*}
Then, if $Y_{k_0:k_0+n+1}$ is induced by attack mode $q$ in the ideal case where $\{w_k\}_k$ is identically zero, $r^q_{k_0}$ is zero. In addition, assuming modes $p$ and $q$ are discernible, as a result of Corollary~\ref{cr:n+1stepID}, one can find $\delta$ and $\epsilon$ such that if the attack with mode $q$ causes a severe impact at time instant $k_0$, i.e., $||\mathbf{z}^{q}_{k_0}(0, \{a_k\}_k, 0)|| \geq \delta$, then $r^p_{k_0} \geq \epsilon$. Hence, the following criteria can be applied to identify the attack, once the attack is detected at $k = 0$ (attack identification is performed only when an attack has been detected): 
\begin{equation}\label{eq:IDlogic}
\begin{aligned}
	\hat Q_0 &= Q\\
	\hat Q_{k+1} &= \hat Q_{k} \setminus \{q | r^q_k \geq \epsilon^q\}
\end{aligned}
\end{equation}
where $\hat Q$ is the set estimate of attack mode.

The following Proposition summarize the effectiveness of (\ref{eq:IDlogic}): 
\\
\begin{proposition}
	Suppose that the system is not vulnerable to any attack mode in $Q$, and $Q$ is identifiable. Assume that $\hat \Sigma$ in (\ref{eq:NomLiftedSys}) is stable, and $||w_k|| \leq 1, \forall k$, without loss of generality. Now consider the case that an alarm for attack detection has been triggered at $k = 0$, and the attack mode is constant. There are $\delta^q > 0$ and $\epsilon^q > 0$, $q \in Q$, such that by applying (\ref{eq:IDlogic}), the following statements hold:
	\begin{enumerate}[(i)]
	\item if the true attack mode is $q$, then $q \in \hat Q_k$ for all $k$;
	\item if the true attack mode is $q$ and the resulted $z_k$ satisfies
	\begin{equation}\label{eq:SeverAttack2BeID}
		||\mathbf{z}^{q}_{k_0}(x_{k_0}, \{a_k\}_k, \{w_k\}_k)|| \geq \delta^q,
	\end{equation}
	then, $\hat Q_{k_0+1} = \{q\}$.
	\end{enumerate}
\end{proposition}

\begin{proof}
For statement (i), we just need to show the existence of $\epsilon^q$ so that $r^q_k < \epsilon^q$ for all $k$. By the condition that $\hat \Sigma$ in (\ref{eq:NomLiftedSys}) is stable and $||w_k||\leq 1$ for all $k$, there is an $\tilde \epsilon$ so that $||\mathbf{y}^{q}_{k}(0, 0, \{w_k\}_{k=0}^\infty)||< \tilde \epsilon$ for all $k$, and thus for any $x_0$, $\{a_k\}_k$, and $\{w_k\}_k$ ($||w_k||\leq 1$ for all $k$), 
\begin{equation}
\begin{aligned}
r^q_{k} &= ||P^\perp_{\mathcal{Y}^q_{n+1}} \begin{bsmallmatrix}
\mathbf{y}^{q}_{k}(x_0, \{a_k\}_k, \{w_k\}_k) \\
\vdots\\
\mathbf{y}^{q}_{k+n+1}(x_0, \{a_k\}_k, \{w_k\}_k)
\end{bsmallmatrix}|| \\
&= ||P^\perp_{\mathcal{Y}^q_{n+1}} \begin{bsmallmatrix}
\mathbf{y}^{q}_{k}(0, 0, \{w_k\}_k) \\
\vdots\\
\mathbf{y}^{q}_{k+n+1}(0, 0, \{w_k\}_k)
\end{bsmallmatrix}|| < \epsilon^q
\end{aligned}	
\end{equation}
for some $\epsilon^q$. 

We now need to show that for any $x^q_{k_0}$, $\{a^q_k\}_k$, and $\{w_k\}_k$, respectively, if $p \neq q$ is in $\hat Q_{k_0+1}$, $||\mathbf{z}^{q}_{k_0}(x^q_{k_0}, \{a^q_k\}_k, \{w_k\}_k)||$ is bounded. $p \neq q$ is in $\hat Q_{k_0+1}$ means that $r^p_{k_0} < \epsilon^p$, which implies that for $Y_{k_0:k_0+n+1}$ given as
\begin{equation}
	Y_{k_0:k_0+n+1} = \begin{bsmallmatrix}
	\mathbf{y}^{q}_{k_0}(x^q_{k_0}, \{a^q_k\}_k, \{w_k\}_k) \\
	\vdots\\
	\mathbf{y}^{q}_{k_0+n+1}(x^q_{k_0}, \{a^q_k\}_k, \{w_k\}_k)
	\end{bsmallmatrix},
\end{equation}
we have
\begin{equation}
\begin{aligned}
r^p_{k_0} =& ||P^\perp_{\mathcal{Y}^p_{n+1}} Y_{k_0:k_0+n+1}|| \\
=& ||Y_{k_0:k_0+n+1} - P_{\mathcal{Y}^p_{n+1}} Y_{k_0:k_0+n+1}||\\
=& ||Y_{k_0:k_0+n+1} - \begin{bsmallmatrix}
\mathbf{y}^{p}_{k_0}(x^p_{k_0}, \{a^p_k\}_k,0) \\
\vdots\\
\mathbf{y}^{p}_{k_0+n+1}(x^p_{k_0}, \{a^p_k\}_k, 0)
\end{bsmallmatrix}|| \\
<& \epsilon^p
\end{aligned}
\end{equation}
for some $x^p_0$, and $\{a^p_k\}_k$. By introducing $x^{pq}_{k_0} = [x^p_{k_0}{}' \,\ x^q_{k_0}{}']'$ and $a^{pq}_{k} = [a^p_{k}{}' \,\ a^q_{k}{}']'$, we have
\begin{equation}
	||y^{pq}_{k}(x^{pq}_{k_0}, \{a^{pq}_k\}_k) || \leq \bar{\epsilon}^{pq}, \forall k \in\{k_0,...,k_0+n+1\}
\end{equation}
for some $\bar{\epsilon}^{pq}$, where we have used that $||\mathbf{y}^{q}_{k}(0, 0, \{w_k\}_k||$ is bounded. By Corollary~\ref{cr:n+1stepID}, it is implied that \begin{equation}
	||\mathbf{z}^{pq}_{k_0}(x^{pq}_{k_0}, \{a^{pq}_k\}_k)|| < \bar \delta^{pq},
\end{equation}
for a $\bar \delta^{pq}$, and thus, $||\mathbf{z}^{q}_{k_0}(x^{q}_{k_0}, \{a^{q}\}_k, \{w_k\}_k||$ is bounded under the stability condition. 

\end{proof}

\begin{remark}
Similar to what we have commented in Remark~\ref{rm:designresidual}, discernibility, as a qualitative property, is invariant under the transformation of the output $y$. By transforming $y$ to $\tilde{y}_\mathtt{z} = H_y(\mathtt{z})y_\mathtt{z}$, where $H_y(\mathtt{z})$ is any stable and invertible filter with stable inverse, the identification scheme could better distinguish attack modes with less sensitivity to noises. It provides certain degrees of freedom for reshaping $y$ and modifying $r^q_k$ in (\ref{eq:defResidual}). 
\end{remark}

\section{An Illustrative Example} \label{sec:NE}
Motivated by \cite{kwon2014analysis}, we consider the following model for describing the two-dimensional navigational behavior of a UAS in the nominal case:
\begin{equation}
	\begin{aligned}
	\hat x_{t+1} &= \underbrace{\begin{bsmallmatrix}
	A_o & -B_oK_o \\ L_oC_o & A_o - B_oK_o - L_o C_o
	\end{bsmallmatrix}}_{\hat A} \hat x_t  + \underbrace{\begin{bsmallmatrix}
	B_oK_o \\ B_oK_o 
	\end{bsmallmatrix}}_{\hat B^u}\hat u_t \\
	\hat y_t &= \underbrace{\begin{bsmallmatrix}
	C_o &0 \\
	0 &I_{4\times4}
	\end{bsmallmatrix}}_{\hat C}\hat x_t
	\end{aligned}
\end{equation}
where the state vector includes the horizontal (north and east) position, velocity, and corresponding observer states; the output used for monitoring the system includes the GPS measurement and the observer state. The kinematic model, guidance control gain, and observer gain are assumed to be the following for demonstration:  
\begin{equation}
\begin{aligned}
A_o &= \begin{bsmallmatrix}
1 &0 &\Delta t &0 \\ 0 &1 &0 &\Delta_t \\ 0 &0 &1 &0 \\ 0 &0 &0 &1
\end{bsmallmatrix}, B_o = \begin{bsmallmatrix}
\frac{\Delta t^2}{2} &0 \\ 0 &\frac{\Delta t^2}{2} \\ \Delta t &0 \\0 &\Delta t
\end{bsmallmatrix}, L_o = \begin{bsmallmatrix}
1.09 &0 \\ 0 &1.09 \\ 0.94 &0 \\ 0 &0.94
\end{bsmallmatrix},\\
C_o &= \begin{bsmallmatrix}
1 &0 &0 &0\\
0 &1 &0 &0
\end{bsmallmatrix}, K_o = \begin{bsmallmatrix}
	9.89 &0 &7.24 &0\\
	0 &9.89 &0 &7.24
	\end{bsmallmatrix}, \Delta t = 0.1.
\end{aligned}
\end{equation}\normalsize
The actual system is subject to disturbance and noise, and we assume the position measurement from the on-board GPS receiver is vulnerable to spoofing attacks. Thus, the deviations of the actual state and output from the nominal state and output have the following relationship:
\begin{equation}\label{eq:vulsys}
\begin{aligned}
	x_{t+1} &= \hat A x_t +  \hat B^a a_t + \hat B^w w_t, \\
	y_t &= \hat C x_t + \hat D^a a_t + \hat D^w w_t, \\
	z_t &= \hat E x_t 
\end{aligned}
\end{equation}
where $a_t = [a^N_t{} \,\ a^E_t{}]'$ is the spoofing attack input to the on-board GPS measurement, and $w_t$ represents the bounded disturbance and noise, 
\begin{equation*}
	\hat B^a = \begin{bsmallmatrix}
	0 \\ L_o
	\end{bsmallmatrix}, \hat{B}^w = \begin{bsmallmatrix}
	B_o &0 \\ 0 &L_o
	\end{bsmallmatrix}, \hat D^a = \begin{bsmallmatrix}
	I_{2\times2} \\0
	\end{bsmallmatrix}, \hat D^w = \begin{bsmallmatrix}
	 0 &I_{2\times2} \\0 &0
	\end{bsmallmatrix},
\end{equation*}
and $\hat E = [I_{2\times2} \,\ 0]$ so that $z_t$ denotes the deviation of the position from the reference trajectory. We can see that $\mathcal{F}(\mathcal{V})$ has only one element for this system: $M = [-I_{2\times2} \,\ 0]$, $N = 0$. As $(\hat A + \hat B^a M)|_{V^*}$ has a non-defective unstable eigenvalue which is $1$, after checking condition (iii) of Theorem~\ref{thm:Detectability}, we can conclude the system is vulnerable. Actually, using the method given in the proof, we can construct a stealthy attack according to (\ref{eq:initialEx}) and (\ref{eq:subcasepolicy}), which can make $||z_t||$ linearly divergent. The simulation results showing the impact of such an attack are given in Fig.~\ref{fig:SR1}.
\begin{figure}[htb]
	\centering
	\includegraphics[trim = 0mm 120mm 80mm 0mm, clip, scale=0.44]{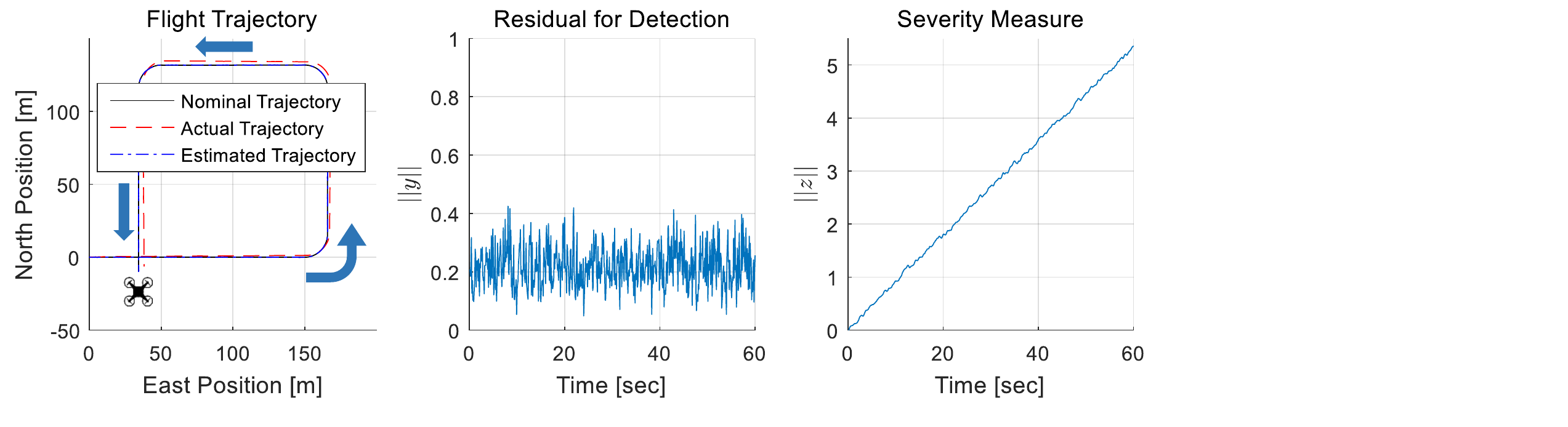}
	\caption{Vulnerability of the System (\ref{eq:vulsys})}
	\label{fig:SR1}
\end{figure}
In the plot, it can be observed that $z_t$ keeps increasing but $y_t$ does not reflect the change of $z_t$. Now, we suppose that there is a secure off-board  measurement (e.g., position measurements by radar systems) at the ground control station, and it is streamed to the UAS every $5$ time steps for on-board monitoring (the streamed data can be delayed as long as the delay is less than $5$ time steps). Due to the multi-rate issue and the potential delay, we consider the lifting technique now. A lifted system with $T = 5$ can be constructed according to Section~\ref{sec:PF}, and we can obtain the system relating $x_k$, $y_k$, and $z_k$ in the form of (\ref{eq:DeltaLiftedSys}). We can generate $y_k$ now from the lifted system, and by Theorem~\ref{thm:Detectability}, the attack can be detected using $||y_k||$. By detectability, we can use Proposition~\ref{pr:detectionwithnoise} to design $\epsilon$ and estimate $\delta$. 
Fig.~\ref{fig:SR2} shows that $||y_k||$ can be used for attack detection and $\epsilon$ can be used as the threshold for triggering the alarm.
\begin{figure}[htb]
	\centering
	\includegraphics[trim = 0mm 60mm 80mm 0mm, clip, scale=0.44]{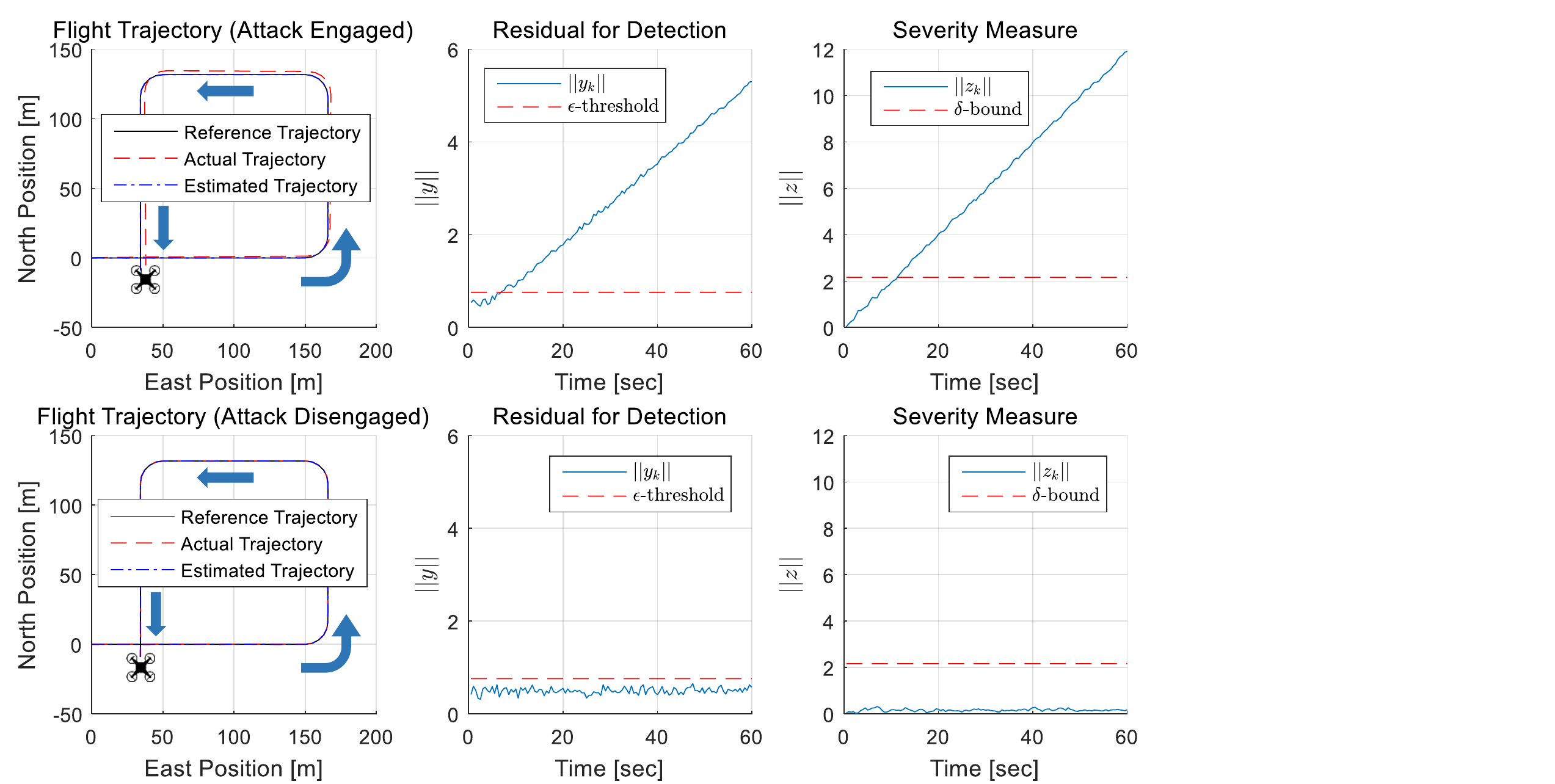}
	\caption{Attack Detection for the Lifted System}
	\label{fig:SR2}
\end{figure}
The previously designed stealthy attack can now be detected: $y_k$ is increasing as $z_k$ is increasing. Also note that the $\epsilon$-threshold avoids the false alarm and the trajectory deviation is below the $\delta$-bound when the attack is disengaged. 

Now, we test the attack mode identification. Let us assume that there are following two attack modes for demonstration: for mode $1$ attack, only the north-position measurement is falsified, and for mode $2$, only the east-position measurement is falsified. Using Theorem~\ref{thm:Discernibility}, we can conclude that the two attack modes are discernible. Fig.~\ref{fig:SR3} shows the $r^1_k$ and $r^2_k$ histories generated using (\ref{eq:defResidual}) when the true attack mode is $1$.
\begin{figure}[htb]
	\centering
	\includegraphics[trim = 0mm 110mm 100mm 0mm, clip, scale=0.55]{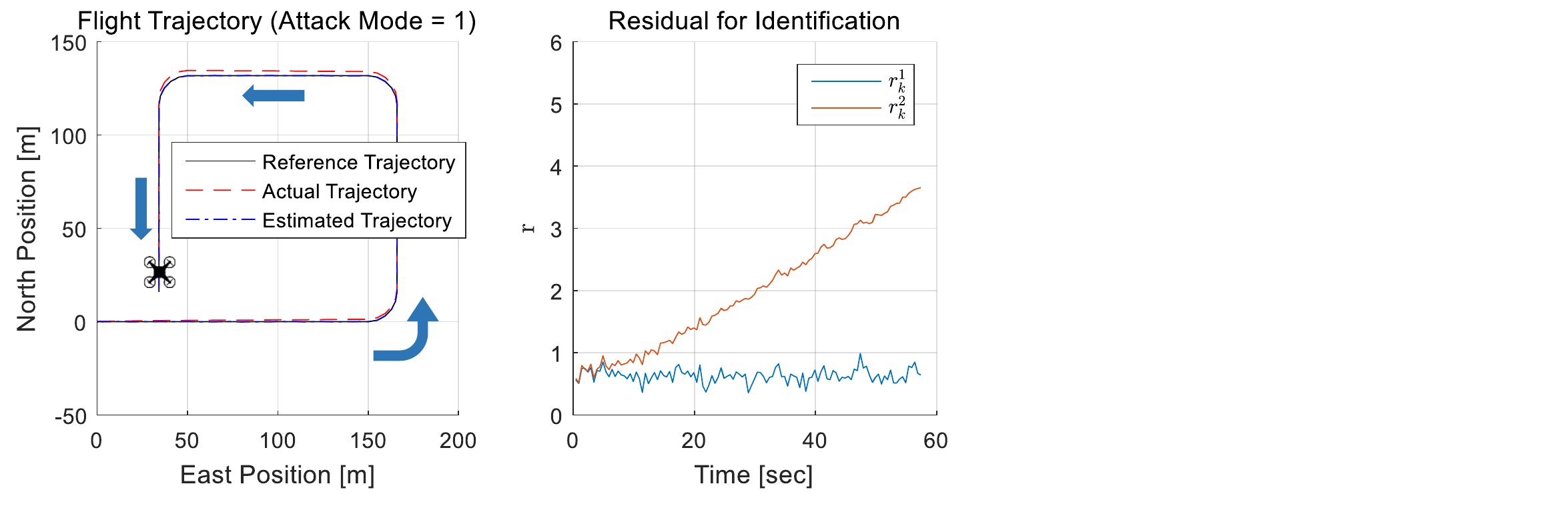}
	\caption{Test of Attack Identification}
	\label{fig:SR3}
\end{figure}
As shown in Fig.~\ref{fig:SR3}, $r^1_k$ remains the same level while $r^2_k$ increases as the trajectory deviation increases, which means identification strategy (\ref{eq:IDlogic}) can provide the correct estimate of the attack mode. 

\section{Conclusion} \label{sec:Conclusion}
In this paper, we have considered attack detection and identification using the lifted system model. The severe attack detectability and identifiability based on the lifted system model have been formally defined and strictly characterized using geometric control theory. A method of checking detectability has been proposed and discussed, and the residual design problem for attack detection has been formulated in a general way. For attack identification, we have discussed it by generalizing the concept of mode discernibility, and a strategy for identifying the attack mode has been derived based on the theoretical analysis. An illustrative example of an unmanned aircraft system has been provided to demonstrate the main results. 

\bibliographystyle{plain}        
\bibliography{RefList}           

\section*{Appendix A}
According to our definition of $\mathcal{F}(\mathcal{V})$, we have
\begin{equation}
\begin{aligned}
\text{Im} \{B^a_{q^a} (\hat M - \bar M)V\} &\subset \mathcal{V}, \\
\text{Im} \{(\hat{M} - \bar M)V\} &\subset \text{ker}\{D^a_{q^a}\},
\end{aligned}
\end{equation}
and it is followed by 
\begin{equation*}
\begin{aligned}
\text{Im}\{B^a_{q^a} (\hat M - \bar M)V\} \subset B^a_{q^a}\text{ker}\{D^a_{q^a}\} \cap \mathcal{V} = \text{Im}\{B^a_{q^a}\hat N\},
\end{aligned}
\end{equation*}
which implies $B^a_{q^a} (\hat M - \bar M)V = B^a_{q^a}\hat N \hat K$ for some $\hat K$. 

Furthermore, we have
\begin{equation}
B^a_{q^a}\left[(\hat M - \bar M)V - \hat N \hat K \right] = 0,
\end{equation}
which implies that $\text{Im}\{(\hat M - \bar M)V - \hat N \hat K\}$ is contained in the null space of $B^a_{q^a}$. In addition, since both $\text{Im} \{(\hat{M} - \bar M)V\}$ and $\text{Im}\{\hat N \}$ are contained in the null space of $D^a_{q^a}$, we have
\begin{equation}
\text{Im}\{(\hat M - \bar M)V - \hat N \hat K\} \subset \text{ker}\{B^a_{q^a}\}\cap \text{ker}\{D^a_{q^a}\},
\end{equation}
where we can conclude the first part of the proof. 

Since $\text{Im}\{B^a_{q^a}\bar N\} = \text{Im}\{B^a_{q^a} \hat N\}$, we have $B^a_{q^a} \bar N = B^a_{q^a} \hat N \hat L$
for some $\hat L$. Then, $\text{Im}\{\bar N - \hat N \hat L\}$ is necessarily contained in the null space of $B^a_{q^a}$. Since both $\text{Im}\{\bar N\}$ and $\text{Im}\{\hat N\}$ are contained in the null space of $D^a_{q^a}$, we have
\begin{equation}
\text{Im}\{\bar N - \hat N \hat L\} \subset \text{ker}\{B^a_{q^a}\}\cap \text{ker}\{D^a_{q^a}\},
\end{equation}
which concludes the second part of the proof.

\section*{Appendix B}
For necessity, if the output $y_k$ is identically zero, then the state is always included in $\mathcal{V}$. By $Ax_k + B^a_{q^a}(a_k + Mx_k- Mx_k) = (A+B^a_{q^a}M)x_k + B^a_{q^a}(a_k - Mx_k) = x_{k+1}$, we can see that $B^a_{q^a}(a_k - Mx_k) \in \mathcal{V}$ for each $k$. By $Cx_k + D^a_{q^a}(a_k + M x_k - Mx_k) = (C+D^a_{q^a}M)x_k + D^a_{q^a}(a_k - Mx_k) = y_k = 0$ for each $k$, we can see that $(a_k - Mx_k) \in \text{ker}\{D^a_{q^a}\}$. Therefore, $a_k - Mx_k$ can be written as $N\tilde{a}_k$ for some $\tilde{a}_k$.

For sufficiency, if $a_k$ takes the form $Mx_k + N\tilde{a}_k$, then $x_k$ can be written as
\begin{equation*}
x_k = (A+B^a_{q^a}M)^k x_ 0 + \Sigma_{j=0}^t (A+B^a_{q^a}M)^{k-j-1}B^a_{q^a}N\tilde{a}_j,
\end{equation*}
which implies $y_k$ is identically zero according to (\ref{eq:VisInv}), (\ref{eq:invDOF}), and $x_0 \in \mathcal{V}$.

\section*{Appendix C}
First, by rewriting Definition \ref{def:Detectability}, we know the system is vulnerable if and only if for any $\delta > 0$ and $\epsilon > 0$, there is an attack input sequence $\{a_k\}_{k=0}^\infty$ such that 
\begin{equation}\label{eq:StealthyAndSevere}
\begin{aligned}
||\mathbf{y}^{q^a}_k(0, \{a_t\}_{t=0}^\infty,0)|| &< \epsilon, \forall k \in \mathbb{N}, \\
||\mathbf{z}^{q^a}_K(0, \{a_t\}_{t=0}^\infty,0)|| &\geq \delta, \text{ for some } K.
\end{aligned}
\end{equation}

We show sufficiency now. Suppose condition (i) holds. Then, clearly there is an attack sequence $\{a_k\}_k$ such that it affects $\{z_k\}_k$, but it results in identically zero $\{x_k\}_k$ and $\{y_k\}_k$ when $x_0 = 0$. By scaling this $\{a_k\}_k$, clearly (\ref{eq:StealthyAndSevere}) can be satisfied.

Suppose condition (ii) holds. Then, one can construct an input sequence by taking the policy in the form $\hat a_k = \hat M x_k + \hat N \tilde{a}_k$. Then, the resulting $\{y_k\}_k$ is identically zero by Lemma \ref{lm:NullingPolicy}. As $L_{(\hat M, \hat N)}^{\Sigma^\Delta}$ is nonzero, $\{\tilde{a}_k\}_k$ can be chosen such that $z_K$ is not zero for some $K$. Then, an attack input sequence $\{a_k\}_k$ can be obtained by scaling $\{\hat a_k\}_k$. According to the linearity of the system, the resulted $\{y_k\}_k$ is still identically zero, while $||z_K||$ can be made arbitrarily large. 

Suppose condition (iii) holds. Let $J(\lambda)$ and $S$ be the Jordan block and the matrix of the corresponding chain of generalized eigenvectors satisfying condition (iii). From (\ref{eq:EigNotKer}), we know $(E+F^a_{q^a}\hat{M})V^*S \neq 0$. Let $i^*$ be the smallest positive integer such that the $i^*$-th column of  $(E+F^a_{q^a}\hat{M})V^*S \neq 0$ is nonzero. For convenience, we shall denote the $i$-th column of $(E+F^a_{q^a}\hat{M})V^*S$ by $\eta_i$. We will consider 3 cases: case 1) $|\lambda| > 1$, case 2) $|\lambda|= 1$ and $i^*$ is strictly less than the number of columns of $S$, and case 3) $|\lambda| = 1$ and $i^*$ is equal to the number of columns of $S$.

For case 1), for any $\epsilon > 0$, we can find an input sequence $\{a_k\}_{k=0}^{n-1}$ such that 
\begin{equation}\label{eq:DesignSmallEx}
||\mathbf{y}^{q^a}_k(0,\{a_k\}_{k=0}^{n-1},0)|| < \epsilon,  \forall t \in \{0,1,...,n-1\}
\end{equation} 
and
\begin{equation}\label{eq:initialEx}
\mathbf{x}^{q^a}_n(0, \{a_k\}_{k = 0}^{n-1},0) = \alpha V^*S e_{i^*}
\end{equation}
for some nonzero real number $\alpha$, where $e_{i^*}$ is a vector whose $i^*$-th element is $1$ and other elements are zeros. It is possible since $\alpha V^*S e_{i^*}$ is contained in the controllable subspace~\footnote{\label{fn:complexEig}We only consider the case where $\lambda$ and $S$ are real without loss of generality. If $\lambda$ and $S$ are complex, $Se_{i^*}$ is replaced with the real part of it or the imaginary of it depending on which part is not contained in the null space of $(E+F^a_{q^a}\hat M)V^*$. }. For $k \geq n$, $a_k$ can be constructed according to the policy $a_k = \hat M x_k$. Then, the resulted output $\{y_k\}_{k = n}^\infty$ will be identically zero by Lemma \ref{lm:NullingPolicy}. The resulted $\{z_t\}_{t=n}^\infty$ satisfies
\begin{equation}
\begin{aligned}
z_k &= (E+F^a_{q^a}\hat M) (A+B^a_{q^a}\hat{M})^{k-n} \alpha V^* S e_{i^*} \\ 
&= \alpha (E+F^a_{q^a}\hat{M})V^* (A+B^a_{q^a}\hat M)|_{V^*}{}^{k-n} S e_{i^*} \\
&= \alpha (E+F^a_{q^a}\hat{M})V^*S J(\lambda)^{k-n} e_{i^*} \\
&= \alpha \begin{bmatrix}
0 &\dots &0 &\eta_{i^*} &*
\end{bmatrix} \begin{bmatrix}
* &\lambda^{k-n} &0 &\dots &0
\end{bmatrix}' \\
&= \alpha\lambda^{k-n} \eta_{i^*}.
\end{aligned}
\end{equation} 
Clearly, for any $\delta > 0$, $||z_K|| > \delta$ for some sufficiently large $K$ as $|\lambda| > 1$.  

For case 2), we can repeat all the steps for case 1) except that we replace $e_{i^*}$ with $e_{i^*+1}$, and then we can ensure that the resulted $||y_k||$ is small for all $k$, while for the resulted $\{z_k\}_k$, we have for $k \geq n$,
\begin{equation}
\begin{aligned}
z_k &= (E+F^a_{q^a}\hat M) (A+B^a_{q^a}\hat{M})^{k-n} \alpha V^* S e_{i^*+1} \\ 
&= \alpha (E+F^a_{q^a}\hat{M})V^* (A+B^a_{q^a}\hat M)|_{V^*}{}^{k-n} S e_{i^*+1} \\
&= \alpha (E+F^a_{q^a}\hat{M})V^*S J(\lambda)^{k-n} e_{i^*+1} \\
&= \alpha \begin{bmatrix}
0 &\dots &0 &\eta_{i^*} &\eta_{i^*+1} &*
\end{bmatrix} \cdot \\
& \,\ \begin{bmatrix}
* &(k-n)\lambda^{k-n-1} &\lambda^{k-n} &0 &\dots &0
\end{bmatrix}' \\
&= \alpha(k-n)\lambda^{k-n-1} \eta_{i^*} + \alpha\lambda^{k-n} \eta_{i^*+1}.
\end{aligned}
\end{equation}
Thus, $||z_k|| \geq (k-n) \alpha ||\eta_{i^*}|| - \alpha||\eta_{i^*+1}||$, can become arbitrarily large as $k$ increases.

For case 3), we construct an attack sequence $\{a_k\}_k$ for all $k \geq n$ using the policy
\begin{equation}\label{eq:subcasepolicy}
	a_k = M x_k + \lambda^{\tilde{k}n}a_j - \lambda^{\tilde{k}n} \hat M x_j
\end{equation} 
where $k = \tilde{k}n+j$, $j = 0, 1, ..., n-1$. We now use induction to prove that under such $\{a_k\}_k$, the corresponding state trajectory at any time instant $k \geq n$ satisfies
\begin{equation}\label{eq:ToshowbyInduction}
	x_{\tilde{k}n+j} = \alpha \tilde{k}\lambda^{\tilde{k}n+j-n}V^*S e_{i^*} + \lambda^{\tilde{k}n}x_j + \tilde{x}_{\tilde{k}n+j}
\end{equation}
where $\tilde{x}_{\tilde{k}n+j}$ is contained in $\text{Im}\{V^*\tilde{S}\}$, and $\tilde{S}$ is $S$ excluding the last column (it is zero if $S$ only has one column). Suppose (\ref{eq:ToshowbyInduction}) is true, then for $j<n-1$,
\begin{equation}
\begin{aligned}
	&x_{\tilde{k}n+j+1} = Ax_{\tilde{k}n+j} + B^{a}_{q^a} a_{\tilde{k}n+j} \\
	=&\alpha\tilde{k}\lambda^{\tilde{k}n+j-n}(A+B^{a}_{q^a}\hat M) V^*S e_{i^*} \\
	& \,\ + \lambda^{\tilde{k}n}(Ax_j + B^{a}_{q^a}a_j) + (A+B^{a}_{q^a}\hat M)\tilde{x}_{\tilde{k}n+j}\\
	=&\alpha\tilde{k}\lambda^{\tilde{k}n+j-n+1} V^*S e_{i^*} + \lambda^{\tilde{k}n}x_{j+1} + \tilde{x}_{\tilde{k}n+j+1}
\end{aligned}
\end{equation}
where $\tilde{x}_{\tilde{k}n+j+1}$ is provided as 
\begin{equation}
	(A+B^{a}_{q^a}\hat M)\tilde{x}_{\tilde{k}n+j+1} + \alpha\tilde{k}\lambda^{\tilde{k}n+j-n+1} V^*S e_{i^*-1},
\end{equation}
which is contained in $\text{Im}\{V^*\tilde{S}\}$ by the Jordan form (we can set $e_{0} = 0$). For $j = n-1$, 
\begin{equation}
\begin{aligned}
&x_{(\tilde{k}+1)n} = Ax_{\tilde{k}n+n-1} + B^{a}_{q^a}a_{\tilde{k}n+n-1} \\
=&\alpha \tilde{k}\lambda^{\tilde{k}n-1}(A+B^{a}_{q^a}\hat M) V^*S e_{i^*} \\
&+ \lambda^{\tilde{k}n}(Ax_{n-1} + B^{a}_{q^a}a_{n-1}) + (A+B^{a}_{q^a}\hat M) \tilde{x}_{\tilde{k}n+n-1}\\
=&\alpha \tilde{k}\lambda^{\tilde{k}n} V^*S e_{i^*} + \alpha \lambda^{\tilde{k}n}V^* S e_{i^*} + \tilde{x}_{\tilde{k}n+n}\\
=& \alpha(\tilde{k}+1)\lambda^{\tilde{k}n} V^*S e_{i^*} + \tilde{x}_{(\tilde{k}+1)n},
\end{aligned}
\end{equation}
where $\tilde{x}_{(\tilde{k}+1)n}$ given as
\begin{equation}
	(A+B^{a}_{q^a}\hat M)\tilde{x}_{\tilde{k}n+n-1} + \alpha\tilde{k}\lambda^{\tilde{k}n} V^*S e_{i^*-1}
\end{equation}
is contained in $\text{Im}\{V^*\tilde{S}\}$. Hence, we can claim (\ref{eq:ToshowbyInduction}) is true. Since for case 3), $(E+F^{a}_{q^a} \hat M)V^*\tilde{S} = 0$ and $|\lambda| = 1$, we have
\begin{equation}
\begin{aligned}
	&z_{\tilde{k}n+j} = E z_{\tilde{k}n+j} + F^{a}_{q^a} a_{\tilde{k}n+j} \\
	=& (E+F^{a}_{q^a}\hat M) x_{\tilde{k}n+j} + F^{a}_{q^a}(\lambda^{\tilde{k}n}a_j - \lambda^{\tilde{k}n} \hat M x_j)\\
	=& \alpha\tilde{k}\lambda^{\tilde{k}n+j-n} \eta_{i^*}  +  \lambda^{\tilde{k}n} z_j,
\end{aligned}
\end{equation}
whose magnitude diverges to infinity as $k$ increases. In addition, 
\begin{equation}
\begin{aligned}
	&y_{\tilde{k}n+j} = C z_{\tilde{k}n+j} + D^{a}_{q^a} a_{\tilde{k}n+j} \\
=& (C+D^{a}_{q^a}\hat M) x_{\tilde{k}n+j} + D^{a}_{q^a}(\lambda^{\tilde{k}n}a_j - \lambda^{\tilde{k}n} \hat M x_j)\\
=& \lambda^{\tilde{k}n} y_j,
\end{aligned}
\end{equation}
whose norm is less than $\epsilon$ by (\ref{eq:DesignSmallEx}). Now, we have concluded the sufficiency.

For necessity, we will show that $||\mathbf{y}^{q^a}_k(0, \{a_k\}_{k=0}^\infty, 0)|| \leq \epsilon$ for all $t \in \mathbb{N}$ can imply that $||\mathbf{z}^{q^a}_k(0, \{a_k\}_{k=0}^\infty,0)|| \leq \bar \delta$ for some $\bar \delta$ under the conditions that
\begin{equation}\label{eq:NecessaryCond1}
	\text{ker}\{B^a_{q^a}\}\cap \text{ker}\{D^a_{q^a}\} \subset \text{ker}\{F^a_{q^a}\}
\end{equation}
and for any $(M, {N})\in \mathcal{F}(\mathcal{V})$, $L_{(M, N)}^{\Sigma^\Delta} = 0$ while any generalized eigenspaces of unstable eigenvalues of $(A+B^a_{q^a} M)|_{V^*}$ are contained in $\text{ker}\{(E+F^a_{q^a} M)V^*\}$. 

If $\{a_k\}_{k=0}^\infty$ is such that $||\mathbf{y}^{q^a}_k(0, \{a_k\}_{k=0}^\infty,0)|| \leq \epsilon$ for all $k \in \mathbb{N}$, the resulted state trajectory $\{x_k\}_k$ satisfies that
\begin{equation}\label{eq:dist(V,xk)}
||P_{\mathcal{V}}^\perp x_k|| \leq \epsilon_1, \forall k \in \mathbb{N}
\end{equation}
for an $\epsilon_1$ depending on the system matrices, where $P_{\mathcal{V}}^\perp$ is the orthogonal projection matrix onto the orthogonal complementary subspace of $\mathcal{V}$. Since $\{x_k\}_k$ is the state trajectory starting from the zero initial state, $x_k$ for every $k$ is staying in the controllable subspace, and thus $||P_{\mathcal{C}}^\perp x_k|| = 0$. Using Lemma \ref{lm:projection}, we have
\begin{equation}\label{eq:dist(Vstar,xk)}
||P_{\mathcal{V}^*}^\perp x_k|| \leq \epsilon_2, \forall k \in \mathbb{N}
\end{equation}
for some $\epsilon_2$. 

Take an arbitrary $(\hat{M},\hat{N}) \in \mathcal{F}(\mathcal{V})$, given the resulted state trajectory $\{x_k\}_k$, we can always rewrite $B^a_{q^a} a_k$ as
\begin{equation}\label{eq:rewriteBa}
B^a_{q^a} a_k = B^a_{q^a} \hat M x_k + B^a_{q^a} \hat{N}\tilde{a}_k + B^a_{q^a} \Delta a_k
\end{equation} 
where $\tilde{a}_k$ is such that
\begin{equation}\label{eq:LSPCost}
B^a_{q^a}\hat{N}\tilde{a}_k = \arg\min||B^a_{q^a} a_k - B^a_{q^a} \hat M x_k - B^a_{q^a} \hat N \tilde{a}_k||.
\end{equation}

Now, we will prove that $||B^a_{q^a} \Delta a_k||$ is bounded when $\{a_k\}_{k=0}^\infty$ satisfies $||\mathbf{y}^{q^a}_k(0, \{a_k\}_{k=0}^\infty,0)|| \leq \epsilon$ for all $k \in \mathbb{N}$. Observe that the resulted state trajectory satisfies
\begin{equation}\label{eq:midstep1}
\begin{aligned}
&P_{\mathcal{V}}^\perp x_{k+1} = P_{\mathcal{V}}^\perp (Ax_k + B^a_{q^a} a_k) \\
&= P_{\mathcal{V}}^\perp (A+B^a_{q^a}\hat M) P_{\mathcal{V}}^\perp x_k + P_{\mathcal{V}}^\perp (A+B^a_{q^a} \hat M) P_{\mathcal{V}} x_k \\
&+ P_{\mathcal{V}}^\perp B^a_{q^a} \hat N \tilde{a}_k +  P_{\mathcal{V}}^\perp B^a_{q^a} \Delta a_k \\
&= P_{\mathcal{V}}^\perp (A+B^a_{q^a} \hat M) P_{\mathcal{V}}^\perp x_k +  P_{\mathcal{V}}^\perp B^a_{q^a} \Delta a_k
\end{aligned}
\end{equation}
By (\ref{eq:dist(V,xk)}), we have
\begin{equation} \label{eq:BDaInVperpisSmall}
||P_\mathcal{V}^\perp B^a_{q^a}\Delta a_k|| \leq \epsilon_1 + ||P_{\mathcal{V}}^\perp (A+B^a_{q^a}\hat M)|| \epsilon_1 .
\end{equation}
We also observe that 
\begin{equation}\label{eq:midstep2}
\begin{aligned}
y_k &= Cx_k + D^a_{q^a} a_k \\
&= (C+D^a_{q^a} \hat M)x_k + D^a_{q^a} \hat N \tilde{a}_k + D^a_{q^a}\Delta a_k\\
&= (C+D^a_{q^a}\hat M)P_\mathcal{V}^\perp x_k + D^a_{q^a} \Delta a_k,
\end{aligned}
\end{equation}
where we know 
\begin{equation} \label{eq:DDaisSmall}
||D^a_{q^a}\Delta a_k|| \leq \epsilon + ||C+D^a_{q^a}\hat{M}||\epsilon_1.
\end{equation} 
By (\ref{eq:BDaInVperpisSmall}), (\ref{eq:DDaisSmall}), and Lemma \ref{lm:projection}, there is a $\bar \beta$ ensuring $||P_{B^a_{q^a}\text{ker}\{D^a_{q^a}\}\cap\mathcal{V}}^\perp B^a_{q^a}\Delta a_k||\leq \bar \beta$ for all $t$. Recall the definition of $\Delta a_k$, $B^a_{q^a}\Delta a_k$ is in the orthogonal complementary subspace of $B^a_{q^a}\text{ker}\{D^a_{q^a}\} \cap \mathcal{V}$ by the optimality of (\ref{eq:LSPCost})
, and thus $||B^a_{q^a}\Delta a_k|| \leq \bar \beta$ for all $k$. 

Then, we will find the bound on $||z_k||$. If we define $\xi_k$ as $P_{\mathcal{V}^*} x_k$, we can observe that
\begin{equation}\label{eq:xiDynamics}
\begin{aligned}
&\xi_{k+1} = P_{\mathcal{V}^*} x_{k+1} = P_{\mathcal{V}^*} (Ax_k + B^a_{q^a} a_k) \\
&= P_{\mathcal{V}^*} (A+B^a_{q^a}\hat M) x_k + P_{\mathcal{V}^*} B^a_{q^a}\hat N \tilde{a}_k + P_{\mathcal{V}^*} B^a_{q^a} \Delta a_k\\
&= (A+B^a_{q^a}\hat M) \xi_k + B^a_{q^a}\hat N \tilde{a}_k + \zeta_k
\end{aligned} 
\end{equation} 
where 
\begin{equation}
	\zeta_k = P_{\mathcal{V}^*} \left[(A+B^a_{q^a}\hat M)P_{\mathcal{V}^*}^\perp x_k + B^a_{q^a}\Delta a_k\right].
\end{equation}
From (\ref{eq:xiDynamics}), we can write $\xi_k$ as
\begin{equation}
\begin{aligned}
\xi_k = &\Sigma_{j = 0}^{k-1} (A+B^a_{q^a}\hat{M})^{k-j-1}B^a_{q^a}\hat{N}\tilde{a}_j \\ &+\Sigma_{j = 0}^{k-1} (A+B^a_{q^a}\hat{M})^{k-j-1}\zeta_j
\end{aligned}
\end{equation}
Therefore, the resulted $\{z_k\}_k$ satisfies that
\begin{equation}\label{eq:ExpandZT}
\begin{aligned}
z_k=& (E+F^a_{q^a}\hat{M})\xi_k + (E+F^a_{q^a}\hat{M})P_{\mathcal{V}^*}^\perp x_k \\
&+ F^a_{q^a}\hat{N} \tilde{a}_k + F^a_{q^a}\Delta a_k \\
=& \Sigma_{j = 0}^{k-1} (E+F^a_{q^a}\hat{M})(A+B^a_{q^a}\hat{M})^{k-j-1}B^a_{q^a}\hat{N}\tilde{a}_j \\
&+ \Sigma_{j=0}^{k-1} (E+F^a_{q^a}\hat{M})(A+B^a_{q^a}\hat{M})^{k-j-1}\zeta_j  \\
&+ (E+F^a_{q^a}\hat{M})P_{\mathcal{V}^*}^\perp x_k + F^a_{q^a}\hat{N}\tilde{a}_k + F^a_{q^a}\Delta a_k
\end{aligned}
\end{equation}
Recall the condition that for any $(M,N) \in \mathcal{F}(\mathcal{V})$, $L_{(M,N)}^{\Sigma^\Delta} = 0$, which further implies
\begin{equation} \label{eq:UsefulCond}
\begin{aligned}
F^a_{q^a}N &= 0 \\
(E+F^a_{q^a}M)(A+B^a_{q^a}M)^j B^a_{q^a}N &= 0 , \forall j \in \mathbb{N}
\end{aligned}
\end{equation}
for any $(M,N) \in \mathcal{F}(\mathcal{V})$ by Caylay-Hamilton theorem. Therefore, the first term and the forth term of $z_k$ in (\ref{eq:ExpandZT}) always vanish no matter what $(\hat M, \hat N)$ we choose. 

The third term of (\ref{eq:ExpandZT}) is bounded according to (\ref{eq:dist(Vstar,xk)}):
\begin{equation}
\begin{aligned}
||(E+F^a_{q^a}\hat{M})P_{\mathcal{V}^*}^\perp x_k|| &\leq ||E+F^a_{q^a}\hat M||||P_{\mathcal{V}^*}^\perp x_k|| \\
&\leq ||E+F^a_{q^a}\hat M|| \epsilon_2.
\end{aligned}
\end{equation}

A bound on the fifth term of (\ref{eq:ExpandZT}) can be obtained by using the bound on $||B^a_{q^a}\Delta a_k||$ and $||D^a_{q^a}\Delta a_k||$ as well as the condition (\ref{eq:NecessaryCond1}) and Lemma \ref{lm:projection}.

Last, let us look at the second term of (\ref{eq:ExpandZT}). Since $\zeta_j$ is contained in $\mathcal{V}^*$, it can be written as $V^*\gamma_t$. As we can see $||\zeta_k||$ is bounded for all $k$, so is $||\gamma_k||$. Then, 
\begin{equation}
\begin{aligned}
&\Sigma_{j=0}^{k-1} (E+F^a_{q^a}\hat{M})(A+B^a_{q^a}\hat{M})^{k-j-1}\zeta_j \\
=&\Sigma_{j=0}^{k-1} (E+F^a_{q^a}\hat{M}) V^* (A+B^a_{q^a}\hat M)|_{V^*} {}^{k-j-1} \gamma_t
\end{aligned}
\end{equation}
Write the $(A+B^a_{q^a}\hat M)|_{V^*}$ in the Jordan form
\begin{equation}
(A+B^a_{q^a}\hat M)|_{V^*} = [G_1 \,\ G_2] \begin{bmatrix}
J_1 &0 \\
0 &J_2
\end{bmatrix} [G_1 \,\ G_2]^{-1}
\end{equation}
such that $J_1$ includes all the Jordan blocks associated with the stable eigenvalues (the eigenvalues whose magnitude are strictly less than one). According to our condition, $(E+F^a_{q^a}\hat{M})V^*G_2 = 0$. Therefore, the second term of (\ref{eq:ExpandZT}) can be further rewritten as
\begin{equation}
\Sigma_{j = 0}^{k-1} [(E+F^a_{q^a}\hat{M})V^* G_1 J_1^{t-j-1} \,\ \,\ 0] [G_1 \,\ G_2]^{-1} \gamma_j
\end{equation}
which has an upper bound for all $k$ because $J_1$ is the Jordan block of stable eigenvalues and $\gamma_j$ is bounded for all $j$.

To summarize, when none of conditions (i)-(iii) holds, $||y_k|| < \epsilon$ for all $k$ implies that there is a $\delta (\epsilon)$ such that $||z_k||< \delta$, which means the system is not vulnerable to attack mode $q^a$. Hence, we have concluded the necessity part of the poof. 

\section*{Appendix D}
To show sufficiency, we just need Lemma \ref{lm:ParA+BM} and see that for any $(M,N) \in \mathcal{F}(\mathcal{V})$, $(E + F^a_{q^a}M)V^*= (E + F^a_{q^a} \hat M)V^*$ when conditions (i) and (ii) of Theorem~\ref{thm:Detectability} do not hold, which can be derived from Lemma \ref{lm:Friends}.

Now we show necessity. First, we should see that there should be uncontrollable (we mean not $((A+B^a_{q^a}\hat M)|_{V^*}, \hat B_N)$-controllable throughout this proof) and unstable eigenvalue. If it is not the case, either there are no unstable eigenvalues or all unstable eigenvalues are controllable. However, when condition (ii) of Theorem~\ref{thm:Detectability} does not hold, every generalized right eigenspace of a controllable eigenvalue of $(A+B^a_{q^a}M)|_{V^*} = (A+B^a_{q^a}\hat M)|_{V^*}+\hat B_N K$ is contained in the controllable subspace (no matter which $K$ is chosen), which is contained in the null space of $(E+F^a_{q^a}\hat M) V^*$. However, it is contradictory to condition (iii) of Theorem~\ref{thm:Detectability}. Hence, condition (iii) of Theorem~\ref{thm:Detectability} holds only if there is a $K$ such that the generalized right eigenspace of an uncontrollable unstable eigenvalue of $(A+B^a_{q^a}\hat M)|_{V^*}+\hat B_N K$ is not contained in $\text{ker}\{(E+F^a_{q^a}\hat M)V^*\}$.

\section*{Appendix E}

First, we should note that by using the augmented system $\Sigma^\Delta_{pq}$ in (\ref{eq:AugSys}), attack modes $p$ and $q$ are indiscernible (NOT discernible) if and only if for any $\epsilon > 0$ and $\delta > 0$, there are $x^{pq}_0$ and $\{a^{pq}_k\}_k$ such that 
\begin{equation}\label{eq:indiscernibleAttack}
\begin{aligned}
	&||\mathbf{y}^{pq}_k(x^{pq}_0, \{a^{pq}_k\}_{k=0}^\infty)|| < \epsilon, \,\ \forall k \in \mathbb{N} \\
	&||\mathbf{z}^{pq}_{k_z}(x^{pq}_0, \{a^{pq}_k\}_{k=0}^\infty)|| \geq \delta, \,\ \text{for some } k_z.
\end{aligned}
\end{equation}

Suppose condition (i) does not hold. It is clear that by setting $x^{pq}_0 = 0$, there is an input sequence $\{a^{pq}_k\}_k$ such that it affects $\{z^{pq}_k\}_k$, but it results in identically zero $\{x^{pq}_k\}_k$ and $\{y^{pq}_k\}_k$. By scaling this $\{a^{pq}_k\}_k$, (\ref{eq:indiscernibleAttack}) can be satisfied. 

Suppose condition (ii) does not hold. Then, there are $(\hat M, \hat N) \in \mathcal{F}(\mathcal{V}_{pq})$, $x^{pq}_0 \in \mathcal{V}_{pq}$ and $\tilde{a}_0$ such that 
\begin{equation}
\begin{aligned}
	x_1 &= (A_{pq}+B^a_{pq}\hat M) x^{pq}_0 + B^a_{pq} \hat N \tilde{a} \in \mathcal{V}_{pq}, \\
	y_0 &= (C_{pq}+D^a_{pq}\hat M)x^{pq}_0 + D^a_{pq} N \tilde{a} = 0, \\
	z_0 &= (E_{pq}+F^a_{pq}\hat M)x^{pq}_0 + F^a_{pq} N \tilde{a} \neq 0.
\end{aligned}
\end{equation}
By the property of $\mathcal{V}_{pq}$ and $x_1 \in \mathcal{V}_{pq}$, there is an input sequence $\{a^{pq}_k\}_{k=1}^\infty$ which makes $\mathbf{y}^{pq}_k(x^{pq}_1, \{a^{pq}_k\}_{k=1}^\infty)$ identically zero. By augmenting $a^{pq}_0 = \hat M x^{pq}_0 + \hat N \tilde{a}$ with $\{a^{pq}_k\}_{k=1}^\infty$ and scaling this $\{a^{pq}_k\}_{k=0}^\infty$, (\ref{eq:indiscernibleAttack}) can be satisfied. Now, we have concluded the necessity. 

For sufficiency, we will show that under conditions (i) and (ii), $||\mathbf{y}^{pq}_k(x^{pq}_0, \{a^{pq}_k\}_{k=0}^\infty)|| < \epsilon$ for all $k$ can imply that there is a $\delta$ such that $||\mathbf{z}^{pq}_k(x^{pq}_0, \{a^{pq}_k\}_{k=0}^\infty)|| < \delta$ for all $k$. The majority of the steps are similar to the proof of Theorem~\ref{thm:Detectability}. By $||\mathbf{y}^{pq}_k(x^{pq}_0, \{a^{pq}_k\}_{k=0}^\infty)|| < \epsilon, \forall k$, there is an $\epsilon_1$ so that
\begin{equation}\label{eq:dist(Vpq,xk)}
||P_{\mathcal{V}_{pq}}^\perp x^{pq}_k|| \leq \epsilon_1, \forall k \in \mathbb{N}.
\end{equation}
Take an arbitrary $(\hat{M},\hat{N}) \in \mathcal{F}(\mathcal{V})$, given the resulted state trajectory $\{x^{pq}_k\}_k$, we can always rewrite $B^a_{pq} a^{pq}_k$ as
\begin{equation}\label{eq:rewriteBpqa}
B^a_{pq} a^{pq}_k = B^a_{pq} \hat M x^{pq}_k + B^a_{pq} \hat{N}\tilde{a}^{pq}_k + B^a_{pq} \Delta a^{pq}_k
\end{equation} 
where $\tilde{a}^{pq}_k$ is such that
\begin{equation}\label{eq:pqLSPCost}
B^a_{pq}\hat{N}\tilde{a}^{pq}_k = \arg\min||B^a_{pq} a^{pq}_k - B^a_{pq} \hat M x^{pq}_k - B^a_{pq} \hat N \tilde{a}^{pq}_k||.
\end{equation}
Apply the same argument that has been applied for (\ref{eq:midstep1})-(\ref{eq:DDaisSmall}), we can obtain bounds on $||B^a_{pq} \Delta a^{pq}_k||$ and $||D^a_{pq} \Delta a^{pq}_k||$. Then, we rewrite the resulted $\{z^{pq}_k\}_k$ as
\begin{equation}\label{eq:expandZpq}
\begin{aligned}
z^{pq}_k=& (E_{pq}+F^a_{pq}\hat{M})P_{\mathcal{V}_{pq}} x_k + (E_{pq}+F^a_{pq}\hat{M})P_{\mathcal{V}_{pq}}^\perp x_k \\
&+ F^a_{pq}\hat{N} \tilde{a}_k + F^a_{q^a}\Delta a_k. \\
\end{aligned}
\end{equation}
By condition (i) and Lemma~\ref{lm:projection}, as well as the boundedness of $||B^a_{pq} \Delta a^{pq}_k||$ and $||D^a_{pq} \Delta a^{pq}_k||$, we know the last term of (\ref{eq:expandZpq}) can be bounded. By condition (ii), the first term and the third term vanish. Last, the second term is bounded by (\ref{eq:dist(Vpq,xk)}), and we can conclude that under conditions (i) and (ii), $||\mathbf{y}^{pq}_k(x^{pq}_0, \{a^{pq}_k\}_{k=0}^\infty)|| < \epsilon$ for all $k$ can imply that the boundedness of $||\mathbf{z}^{pq}_k(x^{pq}_0, \{a^{pq}_k\}_{k=0}^\infty)||$.

\section*{Appendix F}
Statement (ii) implies statement (i), clearly. Thus, we only need to show the other direction. 
If 
\begin{equation}
||\mathbf{y}^{pq}_k(x^{pq}_{k_0}, \{a^{pq}_k\}_{k=k_0}^\infty)|| < \epsilon
\end{equation}
for all $k \in \{k_0, ..., k_0 + n + 1\}$, by the property of $\mathcal{V}_{pq}$, $||P_{\mathcal{V}_{pq}}^\perp x_{k_0}||$ and $||P_{\mathcal{V}_{pq}}^\perp x_{k_0 + 1}||$ are bounded by some $\epsilon_1$. By taking an arbitrary $(\hat{M},\hat{N}) \in \mathcal{F}(\mathcal{V}_{pq})$, we can rewrite $B^a_{pq} a^{pq}_{k_0}$ as
\begin{equation}\label{eq:rewriteBpqa0}
B^a_{pq} a^{pq}_{k_0} = B^a_{pq} \hat M x^{pq}_{k_0} + B^a_{pq} \hat{N}\tilde{a}^{pq}_{k_0} + B^a_{pq} \Delta a^{pq}_{k_0}
\end{equation} 
where $\tilde{a}^{pq}_{k_0}$ is such that
\begin{equation}
B^a_{pq}\hat{N}\tilde{a}^{pq}_{k_0} = \arg\min||B^a_{pq} a^{pq}_{k_0} - B^a_{pq} \hat M x^{pq}_{k_0} - B^a_{pq} \hat N \tilde{a}^{pq}_{k_0}||.
\end{equation}
Apply the similar argument that has been applied for (\ref{eq:midstep1})-(\ref{eq:DDaisSmall}), we can obtain bounds on $||B^a_{pq} \Delta a^{pq}_{k_0}||$ and $||D^a_{pq} \Delta a^{pq}_{k_0}||$. Then, we can observe that
\begin{equation}
\begin{aligned}
z^{pq}_{k_0}=& (E_{pq}+F^a_{pq}\hat{M})P_{\mathcal{V}_{pq}} x_{k_0} + (E_{pq}+F^a_{pq}\hat{M})P_{\mathcal{V}_{pq}}^\perp x_{k_0} \\
&+ F^a_{pq}\hat{N} \tilde{a}_{k_0} + F^a_{q^a}\Delta a_{k_0} \\
\end{aligned}
\end{equation}
is bounded by Theorem~\ref{thm:Discernibility} and the similar argument that has been applied for (\ref{eq:expandZpq}). 
\end{document}